\documentclass[11pt]{article}
\usepackage{graphicx,subfigure}
\usepackage{pifont}

\usepackage{amsfonts,amssymb,amsmath,latexsym,ae,aecompl}
\usepackage{graphics}
\usepackage{latexsym}
\usepackage{enumerate}
\usepackage[noend,linesnumbered,ruled,lined]{algorithm2e}
\usepackage{tikz}
\usepackage{caption}
\usepackage{xspace}

\newcommand{\cA}{{\cal A}}
\newcommand{\cB}{{\cal B}}

\usepackage{a4wide}

\newenvironment{proof}{\noindent {\bf Proof.}}{}

\newtheorem{theorem}{Theorem}[section]

\newtheorem{lemma}{Lemma}[section]

\newtheorem{remark}{Remark}[section]

\newcommand{\thp}{treasure hunt problem}

\newcommand{\colb}[1]{{\it \color{blue}{#1}}}

\usepackage{comment}

\begin{document}

\bibliographystyle{plain}

\title{{\bf Pebble Guided Near Optimal Treasure Hunt in Anonymous Graphs} \footnotemark[1]}

%
%
\author{Barun Gorain\footnotemark[2]
\and
Kaushik Mondal\footnotemark[3] \and
Himadri Nayak\footnotemark[4]
\and
Supantha Pandit\footnotemark[5]
}
%
\def\thefootnote{\fnsymbol{footnote}}

\footnotetext[1]{
\noindent
Preliminary version of this paper is accepted in SIROCCO 2021.
}
\footnotetext[2]{
\noindent
Indian Institute of Technology Bhilai, Raipur, Chattisgarh, India. {\tt barun@iitbhilai.ac.in}
}
\footnotetext[3]{
\noindent
Indian Institute of Technology Ropar, Rupnagar, Punjab, India. {\tt kaushik.mondal@iitrpr.ac.in}
}

\footnotetext[4]{
\noindent
Indian Institute of Information Technology, Bhagalpur, India {\tt himadri@iiitbhagalpur.ac.in}
}
\footnotetext[5]{
\noindent
Dhirubhai Ambani Institute of Information and Communication Technology, Gandhinagar, Gujrat, India. {\tt pantha.pandit@gmail.com}
}

\maketitle              
\begin{abstract}
We study the problem of treasure hunt in a graph by a mobile agent. The nodes in the graph are anonymous and the edges at any node $v$ of degree $deg(v)$ are labeled arbitrarily as $0,1,\ldots, deg(v)-1$. A mobile agent, starting from a node, must find a stationary object, called {\it treasure} that is located on an unknown node at a distance $D$ from its initial position. The agent  finds  the treasure when it reaches the node where the treasure is present. The {\it time} of treasure hunt is defined as the number of edges the agent visits before it finds the treasure. The agent does not have any prior knowledge about the graph or the position of the treasure. An Oracle, that knows the graph, the initial position of the agent, and the position of the treasure, places some pebbles on the nodes, at most one per node, of the graph to guide the agent towards the treasure.

We target to answer the question: what is the fastest possible treasure hunt algorithm regardless of the number of pebbles are placed?

We show an algorithm that uses $O(D \log \Delta)$ pebbles to find the treasure in a graph $G$ in time $O(D \log \Delta + \log^3 \Delta)$, where $\Delta$ is the maximum degree of a node in $G$ and $D$ is the distance from the initial position of the agent to the treasure. We show an almost matching lower bound of $\Omega(D \log \Delta)$ on time of the treasure hunt using any number of pebbles.


\end{abstract}

\section{Introduction}

\subsection{Model and Problem Definition}  Treasure hunt by a mobile agent is a well studied problem in networks and related areas. A mobile agent, starting from an initial position, has to find a stationary object, called treasure. In practice, a treasure can be a missing person in a dark cave and a mobile robot must find the person. In networks applications, a software agent must find  a computer virus or valuable data resource in a computer connected in a network.

The network is modeled as a graph where the nodes are unlabeled. The edges at a node $v$ of degree $deg(v)$ are labeled as $0,1,\ldots,deg(v)-1$ arbitrarily. Thus, each edge has two port numbers associated with it at each of its incident vertices. A mobile agent, starting from a node, must find the treasure which is situated in an unknown node at distance $D$. The agent have no prior knowledge  about the network or the value of $D$. The agent finds the existence of the treasure only when it reaches the node where the treasure is situated. The agent moves according to a deterministic algorithm where at each node, it chooses a port and move to the next node using the chosen port. At the start, the agent only knows the degree of the initial node. From a node $u$, when the agent reaches node $v$ by using the port $p$ at $u$, it learns the degree of the node $v$, and the port $q$ at $v$ through which it reaches $v$. For a network with maximum degree $\Delta$, using a simple depth first search based algorithm, the agent can find the treasure in time $O(\Delta^D)$. But many practical applications required a faster treasure hunt algorithm. For example, consider the application of finding a person inside a mine by a mobile robot. The person may be lost and injured due to a sudden accident and therefore he or she must be found as fast as possible. In such scenarios, some external help is provided to the robot/agent in order to guide it towards the desired location faster.
Providing such external help is often done by an oracle which gives some additional information in the form of a binary string, called {\it advice} \cite{GorainP19,Avery2014}. The oracle gives this advice to the agent a priori. Using the information provided in the advice, the agent finds the treasure.
However, if the treasure hunt has to be done  by sufficiently large number of agents independently over a long duration, providing such advice to each of the agents might be costly. Instead, the same purpose can be served by providing information to the nodes of the networks only once. Each agent, while visiting a node, learns this information and find the treasure using it. In this paper, we consider a very simple scenario where a pebble can be placed at the nodes as an external information. An agent, looking at the placement of the pebbles, learns in which way it must traverse to find the treasure. To be specific, we consider the problem where some pebbles are placed  in some nodes of the network by an oracle who knows the initial position of the agent and the position of the treasure. The position of the pebbles guides the agent towards the treasure. At any node, at most one pebble can be placed. The agent can see a pebble only after reaching that node.
In this paper, we study what is the fastest possible algorithm for treasure hunt in anonymous graphs with pebbles. To be specific,
we aim design the fastest algorithm for treasure hunt when any number of pebbles can be placed in the network.

\subsection{Our Result}

\begin{itemize}
    \item[\ding{228}] We present an algorithm that finds the treasure in an anonymous graph in  $O(D \log \Delta+\log ^3 \Delta)$-time using $O(D \log \Delta)$ pebbles, where $\Delta$ is the maximum degree of a node in the graph and  $D$ is the distance from the initial position of the agent to the treasure.


\item[\ding{228}] We prove that even if we supply any number of pebbles, any algorithm must require $\Omega(D \log \Delta)$-time to find the treasure in an anonymous graph.
\end{itemize}

\subsection{Related Work}

Treasure hunt by a mobile agent is a well studied problem \cite{Beck1970,Bouchard20,Bouchard2020,Demaine2006,Emek2015,Kao1996,Avery2014,Amnon2014} for last few decades. In \cite{Beck1970}, Beck et al. introduced the problem of deterministic treasure hunt on a line. The authors proposed a deterministic algorithm with competitive ratio 9 and proved that this ratio is the best that can be acheived in case of line. A generalized version of \cite{Beck1970} was studied by Demaine et al. \cite{Demaine2006}
by considering cost of turns that agent makes along with the cost of the trajectory. Bouchard et al. \cite{Bouchard2020} considered the treasure hunt problem in plane and showed a much improved bound with the assumption of angle information.

In the book \cite{Alpern2003}, several problems related to treasure hunt are discussed. Most of the algorithms surveyed in this book are randomized. One such is the randomized treasure hunt in a star, where the treasure is present in one of the $m$ rays passing through a common point \cite{Kao1996}. In \cite{Avery2014,Amnon2014}, it is shown that the problem of treasure hunt and the problem of rendezvous in graphs are closely related. Ricardo et al. studied the problem of finding an unknown fixed point on a line and in a grid \cite{Ricardo1993}. More generalized studies are done in \cite{Artur2009,Langetepe2012}, where the objective is to search an unknown line in a plane. The author studied the problem of finding a target in a ring in \cite{Spieser2012} by multiple selfish agents participates and a game theoretic solution is proposed. Also, treasure hunt in a plane and in a grid by multiple agents are studied in \cite{Emek2015,Fricke2016,Langer2015}. In \cite{Emek2015,Langer2015}, the agents are considered to have bounded memory. Treasure hunt in plane is studied in \cite {Pelc19arx1} in the advice model. Treasure hunt in a tree network is studied in \cite{Lucas2018}, where random faulty hints are provided to the agents. Treasure hunt in arbitrary graph is also studied in \cite{Bouchard20} considering the agent has unlimited memory. The game of pursuit-evasion, a closely related problem to treasure hunt, is considered in \cite{Bonato2011,Chung2011}, where set of pursuers try to catch a fugitive trying to escape. Also treasure hunt in terrain in presence of obstacles was introduced in \cite{Pelc18arx2}. Another related problem is graph exploration and Disser et. al. \cite{Disser19} recently proved a tight bound on number of pebbles required for a single mobile agent with constant memory to explore an undirected port labelled graph.

\section{Treasure Hunt Algorithm}

In this section, we provide an $O(D \log \Delta+\log^3 \Delta)$-time algorithm for the \thp~using $O(D\log \Delta)$  pebbles.

Let $G$ be a graph with maximum degree $\Delta \ge 2^{10}$. If $\Delta<2^{10}$, then the algorithm described for the case when all nodes on the path from $s$ to $t$ are of `small' degree (these small degree nodes are defined as light nodes later) can be applied.

Let $s,t \in G$ be the starting point of the agent, and position of the treasure in $G$, respectively. Let $P$ be the shortest path between $s$ and $t$ of length $D$. Without loss of generality, we assume that the degree of the node $s$ is at least 2. Otherwise, the first degree-3 node along $P$ starting from $s$ can be considered as the starting position of the agent.
For any node $v \in G$, by $deg(v)$, we denote the degree of the node $v$. Let $\alpha_v=1+ \lfloor \log deg(v)\rfloor$. For any node $w \in G$, by $w(0), w(1), \ldots, w(deg(w)-1)$, we denotes the neighbors of $w$ that are connected through port numbered $0,1, \ldots, deg(w)-1 $, at $w$ respectively.
For any two strings $\Gamma_1$ and $\Gamma_2$, by `$\Gamma_1\cdot \Gamma_2$', we mean concatenation of $\Gamma_1$ and $\Gamma_2$. For any binary string $\Gamma$, by $\Gamma(i,j)$, we denote the substring of $\Gamma$ starting from the $i$-th bit of $\Gamma$ to the $j$-th bit of $\Gamma$. For two nodes $u,v \in G$, we denote the shortest distance between $u$ and $v$ in $G$ by $dist(u,v)$.

Before providing the formal description of the algorithm, we give a high level idea of the pebble placement and discuss how these pebbles guide the agent towards the treasure. To help the reader understand the algorithm better, we describe the idea for trees first and then generalize this idea for general graphs.


\noindent {\bf High level idea of the algorithm in a tree network:} Here we assume that $G$ is a rooted tree with root $s$.
Let $L_i$ be the set of nodes that are at distance $i$ from $s$. Let $ P=(s=)~v_0,v_1, \ldots, v_{D-1}, v_D~(=t)$ be the path from $s$ to $t$, and let $p_0,p_1, \ldots, p_{D-1}$ be the sequence of port numbers corresponding to the path $P$ such that from the node $v_\ell$, the node $v_{\ell+1}$ can be reached by taking the edge with port number $p_\ell$.
For any node $v \in L_j$, let its $i$-th neighbor be the adjacent node in $L_{j+1}$ to which $v$ is connected via $i$-th largest port going to $L_{j+1}$. The pebbles are placed at the children of the nodes in $P$ such that the placement of the pebbles corresponds the binary representation of the port numbers along the shortest path from the current node. To be more specific, let $b_0b_1 \ldots b_{m-1}$ be the binary representation of the integer $p_j$, where $m= 1+ \lfloor\log deg(v_j) \rfloor$. For $0 \le i \le m-1$, place a pebble at $i$-th neighbor of $v_j$, if $b_i=1$. Hence, from the point of view of a node $v_j \in L_i$, each of its neighbors in $L_{i+1}$ either contains a pebble or does not contain a pebble. Visiting each of the neighbor in increasing order of the port numbers and ignoring the port that connects to the parent of the current node, the agent, from the current node, can learn the binary representation of the integer $p_j$, by realizing a node with pebble as `1' and a node without pebble as `0'. Hence, the placement of pebbles in the above manner helps us to ``\colb{encode}" the port labeled path from $s$ to $t$ and the agent, with the help of this encoding, learn the port numbers from each node in $P$ that leads to the next node towards the treasure. The difficulty here for the agent is while learning this binary encoding by looking at the pebble placement, the agent must learn when this binary encoding ends, as the nodes, which are not used for encoding (no pebbles are placed on these nodes) can be misinterpreted as zeros. To overcome this difficulty, instead of simple binary encoding. we use a \colb{transformed binary encoding}: replacing every ‘1’ by ‘11’ and every ‘0’ by’ `10' in the standard binary encoding.  The  advantage  of  this  transformed encoding  is  that  it  does  not contain the substring `00'. Hence, as soon as the agent sees two consecutive node without pebbles, it realizes that the binary string that is encoded in the nodes is ended.

\noindent {\bf Extending the idea for general graphs:} The above method of pebble placement for trees can not be directly extended to general graphs. This is because in a rooted tree, no two nodes have common children. Hence, once the encoding is done after placing the pebbles, the agent can unambiguously decode this encoding. However, in the case of graphs, two consecutive nodes on $P$ may have neighbors in common. Hence, if a pebble is placed on such a common neighbor, the agent can not distinguish for which node the pebble is placed. Also, since the nodes are anonymous, there is no way the agent can identify whether it is visiting a node that is a common neighbor to the previous or next node. We resolve this difficulty by encoding in the neighbors of a set of ``\colb{high}" degree nodes that do not share neighbors. These nodes are called \colb{milestones} (we define this formally later).  The details are  explained below.

We say that a node is \colb{heavy} if its degree is at least $80\lfloor\log \Delta\rfloor+106$ (the reason for choosing this magical number is discussed in Remark \ref{rem:1}).  Otherwise, we say that the node is \colb{light}. Let $T$ be the \colb{breadth first search (BFS)}  tree, rooted at $s$, and for $0\le i\le D$, let $L_i$ denotes the set of nodes that are at a distance $i$ from $s$. Clearly, $s \in L_0$, and $t \in L_D$.

Let $P=(s=)~v_0,,v_1, \ldots, v_{D-1}, v_D~(=t)$, be the shortest path from $s$ to $t$, and let $p_0,p_1, \ldots, p_{D-1}$ be the sequence of port numbers corresponding to the path $P$ such that from the node $v_\ell$, the node $v_{\ell+1}$ can be reached by taking the edge with port number $p_\ell$, for $0 \le \ell<D$. For $0 \le i \le D-1$, let $B_i$ be the binary representation of the integer $p_i$ of length $x_i$, where $x_i=1+ \lfloor \log deg(v_i)\rfloor$.

First, consider a special case where each $v_i$, $0\le i \le D-1$, is light. A  simple algorithm will work in this case: place a pebble at each of the nodes $v_i$, for $1\le i \le D-1$. The agent, at the starting node $s$, set $CurrentNode=s$. At each step, the agent visits all the neighbors of the $CurrentNode$ and move to the neighbor $v$, (except the node from where it reaches to $CurrentNode$) that contains a pebble. It then sets $CurrentNode=v$. The agent continues to explore in this way until the treasure is found. Since all the nodes on $P$ are light, the time for treasure hunt is $O(D \log \Delta)$. If not all nodes are light on $P$, then a set of nodes called \colb{milestones} are used to code the sequence of port numbers corresponding to the path $P$. We define a set of nodes as milestones in a recursive manner. To define the first milestone, we consider the following four cases based on the position of the heavy nodes in the BFS tree $T$: i)  the node $s$ is heavy, ii) the node $s$ is light and a node in $L_1$ is heavy, iii) all the nodes in $L_0 \cup L_1$ are light and a node in $L_2$ is heavy, and iv) all the nodes in $L_0 \cup L_1 \cup L_2$ are light and $j \ge 3$ is the smallest integer for which $L_j$ contains a heavy node that is at a distance 3 from $v_{j-3}$. From now onward, we distinguish these four cases as Case H, Case L-H, Case L-L-H, and Case L-L-L, respectively.

The first milestone is defined based on the above four cases as follows.

\begin{itemize}
\item{\bf Case H:}   $milestone_1=s$.
\item{\bf Case L-H:}  Let $w$ be the node with maximum degree in $L_1$. If multiple nodes with maximum degree exists, then $w$ is chosen as the node to which $s$ is connected by the edge with minimum port number. Set $milestone_1=w$.
\item{\bf Case L-L-H:}  Let $w$ be the node with maximum degree in $L_2$. If multiple nodes with maximum degree exists, then $w$ is chosen as the node to which $s$ is connected by a path of length 2 which is lexicographically shortest among all other paths to  nodes with same degree. Set $milestone_1=w$.
\item{\bf Case L-L-L:}
If $j=4$, and $v_1 \in \{s(0),s(1)\}$, then consider the set  $W=\{ w \in L_{j}| dist(s(0),w)=3 ~or~ dist(s(1),w)=3\}$.  Let $w$ be the node in $W$ whose degree is maximum among all nodes in $W$.  If multiple nodes with maximum degree exists, then the node with maximum degree to which $s$ is connected via  lexicographically shortest path of length 4 is chosen as $w$. Set $milestone_1=w$. On the other hand, If $j>4$ or $v_1 \not \in \{s(0),s(1)\}$, let  $W=\{ w \in L_{j} |dist(v_{j-3},w)=3\}$. Let $w$ be the node with maximum degree in $W$. If multiple nodes with maximum degree exists then $w$ is chosen as the node to which $v_{j-3}$ is connected by a path of length 3 that is lexicographically shortest among all other nodes with same degree. Set $milestone_1=w$.
\end{itemize}

The subsequent milestones are defined recursively as follows. For $i \ge 1$, let $milestone_i \in L_j$. The first heavy node on $P$ that is at a distance at least 3 from $v_j$ is defined as $milestone_{i+1}$.

Intuitively, we encode in the neighbors of the milestones in a similar fashion as described above for the case of tree network. As the milestones are at least 3 distance apart, no two of them have common neighbors and hence decoding can be done unambiguously. However, there is another difficulty, that is the agent does not have any knowledge about the graph and hence does not know the value of $\Delta$, the maximum degree of the graph. This restrict the agent to learn whether a node is heavy or light just by looking at its degree. We overcome this difficulty by placing some `\colb{markers}'. By looking at these markers the agent can identify the possible position of the first milestone. Once the agent identify the first milestone, finding the other milestones are easy as the path towards the next milestone is carefully coded at the neighbors of the current milestone. If $s$ itself is a heavy node, then it is defined as the first milestone. In order to help the agent to learn that this is the case, two pebbles are placed at $s(0)$, and $s(1)$, one at each. The agent, at the beginning of the treasure hunt algorithm, first visits the first two neighbors of $s$, and learn that $s$ is heavy, if it sees pebbles at both of these nodes. If $s$ is light and a node in $L_1$ is heavy, then a pebble is placed at $s(0)$. While vising first two neighbors of $s$, the agent identify the Case L-H  by finding a pebble at $s(0)$, and not finding any pebble at $s(1)$. Once the agent identify the Case L-H, it reach to the node connected to $s$ with maximum degree and hence reach to the first milestone in this case. For Case L-L-H, a pebble is placed at $s(1)$. In similar fashion as earlier, the agent identify the Case L-L-H, by finding a pebble at $s(1)$, and not finding any pebble at $s(0)$. It then reached to the milestone by exploring all paths of length 2 from $s$ and finding the node with maximum degree. For the Case L-L-L, no pebbles are placed in either of the first two neighbors of $s$ and the agent can identify  this case by visiting $s(0)$ and $s(1)$ and not finding any pebbles in these nodes. We explain how the agent identify the first milestone in this case later where formal description of the algorithm is provided.

Another set of `markers' are used to indicate the distance between two consecutive milestones. How these markers are placed will be explained later where formal descriptions of the pebble placements is provided.

We are ready to give the formal description of the treasure hunt algorithm.

\subsubsection{Pebble Placement:}
The placement of pebbles are done in three phases.

\noindent {\boldmath{\bf Phase 1: Placing pebbles  in $s(0)$ and $s(1)$}}
\begin{itemize}
\item {\bf Case H:} Place one pebble each at $s(0)$, and $s(1)$.
\item{\bf Case L-H:} Place a pebble at $s(0)$.
\item{\bf Case L-L-H:} Place a pebble at $s(1)$.
\item{\bf Case L-L-L:} No pebble is placed in $s(0)$, and $s(1)$ is this case.
\end{itemize}

\noindent {\boldmath{\bf Phase 2: Placing pebbles to encode the path between $milestone_1$ to $milestone_2$}}

\begin{description}

\item {\bf Case H:}  Here $milestone_1$ is the node $s$ itself. Notice that, two pebbles are already placed at $s(0)$ and $s(1)$ during phase 1 to represent the marker corresponding to Case H. The other neighbors of $s$ are used to encode the path say $P'$ from $s$ to $milestone_2$. To be more specific, if the distance from $s$ to $milestone_2$ is at most 5, then the entire path $P'$ is coded in the neighbors of $s$. Otherwise, if the distance from $s$ to $milestone_2$ is more than 5, then the first three port numbers and the last three port numbers are encoded at the neighbors of $s$. The difficulty here is to make the agent learn that how far $milestone_2$ is from $s$. To overcome this situation, the third neighbor $s(2)$ and the fourth neighbor $s(3)$ of $s$ are used  as `markers' to represent the distance between $s$ to $milestone_2$. To denote that $milestone_2$ is at distance 3, two pebbles are placed on $s(2)$ and $s(3)$ to encode the marker `11'. To denote $milestone_2$ is at distance 4, a pebble is placed on $s(2)$ and no pebble is placed on $s(3)$ to encode the marker `10'. To denote $milestone_2$ is at distance 5, a pebble is placed on $s(3)$ and no pebble is placed on $s(2)$ to encode the marker `01'.  Finally, no pebbles are placed either on $s(2)$ or on $s(3)$ to encode the marker `00' that indicates that $milestone_2$ is at least 6 distance apart from $milestone_1$.
Note that in this case, ($milestone_2$ is at distance more than 5 from $milestone_1$) the entire path from $milestone_1$ to $milestone_2$ is not coded. Instead, the first three and the last three sequence of port numbers are coded in the neighbors of $milestone_1$. Once the agent learns this coding, can compute the first three ports and
The light nodes in between are used to guide the agent towards $milestone_2$.


Once the pebble corresponding to markers are placed, a few other neighbors of $s$ are used to encode the sequence of port numbers following which the agent reaches  $milestone_2$ from $milestone_1$.
The formal pebbles placement is described based on the distance of the second milestone as follows. Let $\ell$ be an integer such that $milestone_2=v_\ell$.

\begin{enumerate}
\item $\left[\ell=3\right]$: Here the path from $milestone_1$ to $milestone_2$ is $v_0,v_1,v_2,v_3$ and the corresponding sequence of port numbers is $p_0,p_1,p_2$. Let $\Gamma'= B_0 \cdot B_1 \cdot B_2$ and $\Gamma= 11 \cdot \Gamma'$.

\item $\left[\ell=4\right]$: Here the path from $milestone_1$ to $milestone_2$ is $v_0,v_1,v_2,v_3,v_4$ and the corresponding sequence of port numbers is $p_0,p_1,p_2,p_3$. Let $\Gamma'= B_0 \cdot B_1 \cdot B_2 \cdot B_3$ and $\Gamma= 10 \cdot \Gamma'$.

\item $\left[\ell=5\right]$: Here the path from $milestone_1$ to $milestone_2$ is $v_0,v_1,v_2,v_3,v_4,v_5$ and the corresponding sequence of port numbers is $p_0,p_1,p_2,p_3,p_4$. Let $\Gamma'= B_0 \cdot B_1 \cdot B_2 \cdot B_3 \cdot B_4$ and $\Gamma= 01 \cdot \Gamma'$.

\item $\left[\ell \ge 6\right]$:  Here the path from $milestone_1$ to $milestone_2$ is $v_0,v_1, \ldots v_\ell$ and the corresponding sequence of port numbers is $p_0,p_1,\ldots p_\ell$. As mentioned earlier, the sequence of port numbers $p_0,p_1,p_2,p_{\ell-3},p_{\ell-2},p_{\ell-1}$ is coded in this case.
Let $\Gamma'= B_0 \cdot B_1 \cdot B_2 \cdot B_{\ell-3} \cdot B_{\ell-2} \cdot B_{\ell-1}$ and  $\Gamma= 00 \cdot \Gamma'$.
\end{enumerate}
Let $\hat{\Gamma}$ be the transformed binary encoding of $\Gamma$.
Let $z$ be the length of the string $\hat{\Gamma}$.
For $1\leq i\leq z$, place a pebble on $s(1+i)$, if the $i$-th bit of $\hat{\Gamma}$ is 1.
If $\ell \ge 7$, then place a pebble on each of the nodes $v_{4}, \ldots, v_{\ell-3}$.

Figure \ref{fig-h} shows the pebble placement for Case H.

\begin{figure}[ht!]
\begin{center}
{\subfigure[ ]{\includegraphics[scale=.45]{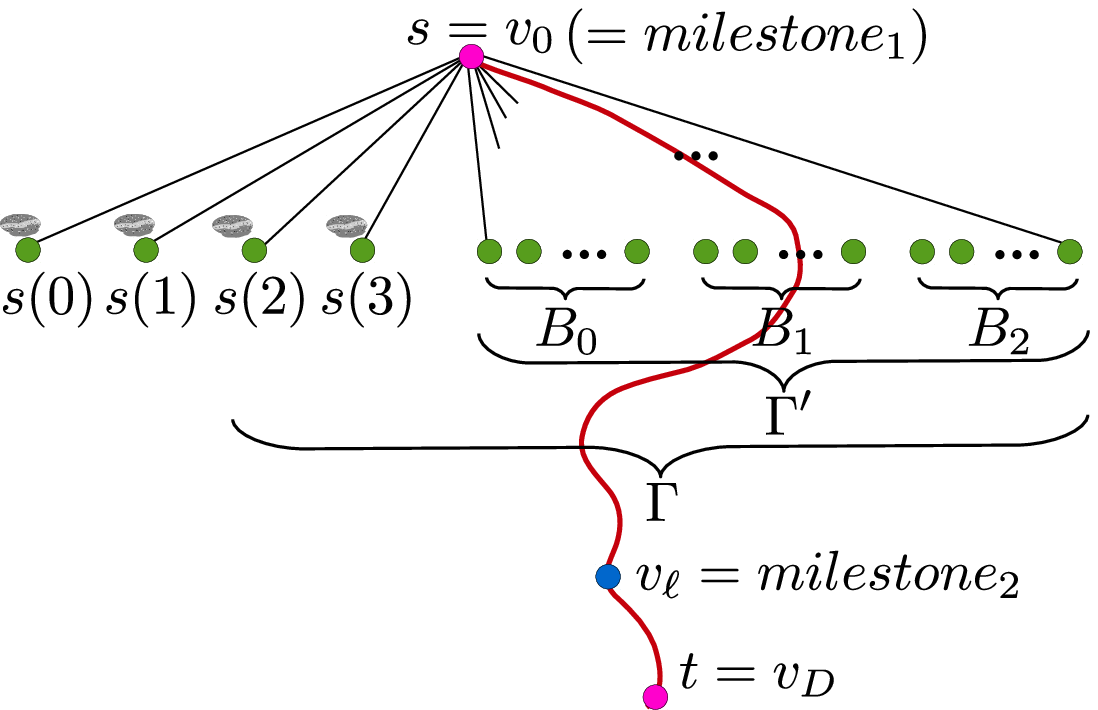}
\label{fig-h-1}
}}
\hspace{.0cm}
{\subfigure[ ]{\includegraphics[scale=.45]{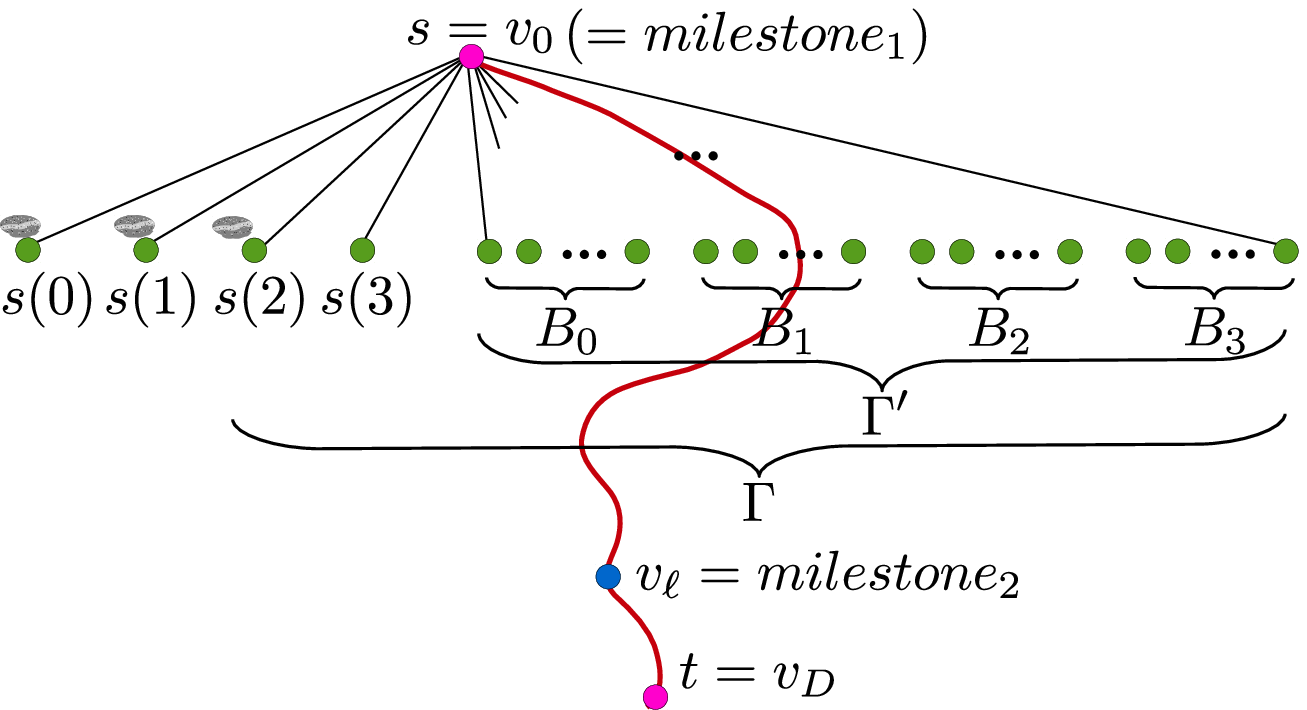}
\label{fig-h-2}
}}

{\subfigure[ ]{\includegraphics[scale=.45]{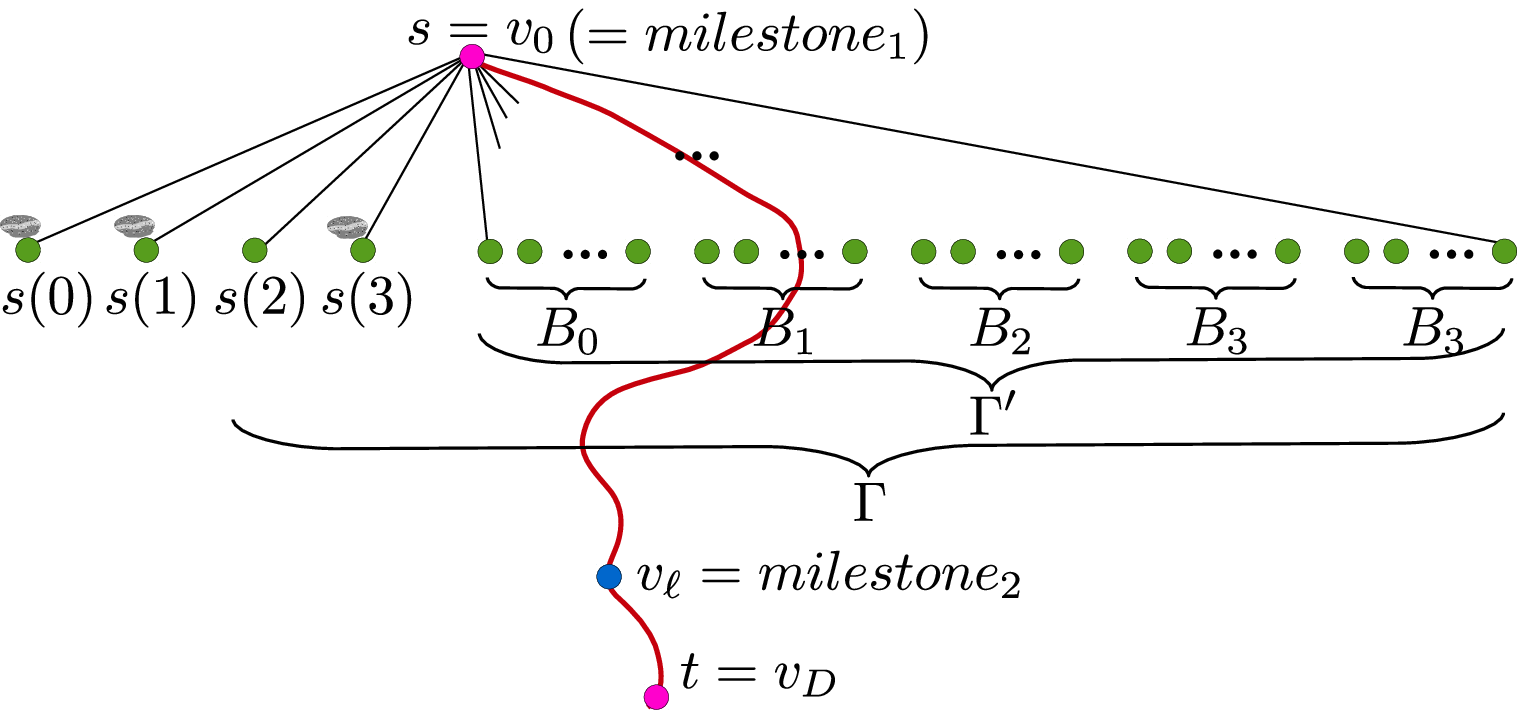}
\label{fig-h-3}
}}
\hspace{.1cm}
{\subfigure[ ]{\includegraphics[scale=.45]{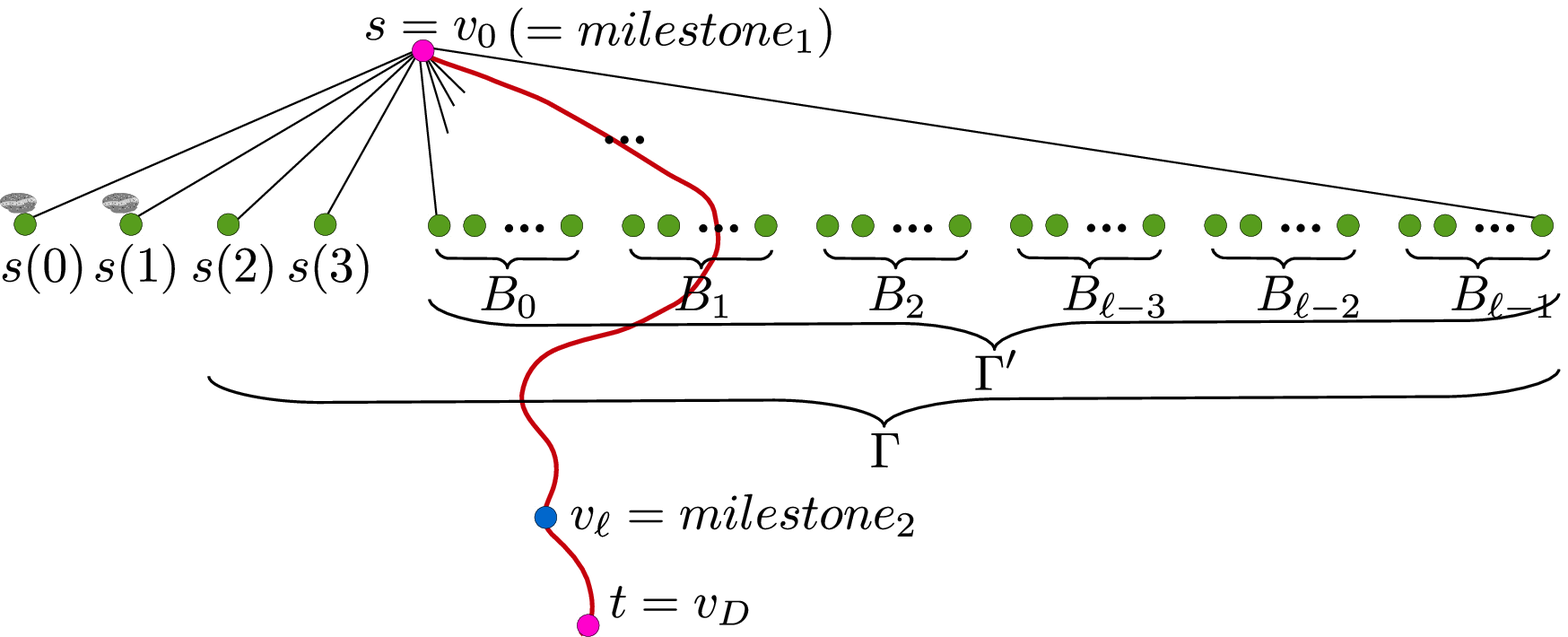}
\label{fig-h-4}
}}
\end{center}
\caption{Showing the pebble placement in Case H. The markers are represented by showing pebbles using gray colored circles above the nodes (a) Placement of pebbles where $milestone_2 \in L_3$ (b) Placement of pebbles where $milestone_2 \in L_4$ (c) Placement of pebbles where $milestone_2 \in L_5$ (d) Placement of pebbles where $milestone_2 \in L_j$ for $j \ge 6$.}
\label{fig-h}
\end{figure}

\item {\bf Case L-H:} In this case, the first milestone is selected as one of the neighbors $w$ of $s$. Note that $w$ may not be on the shortest path from $s$ to $t$. For this reason we encode the path (or a subpath of the path) from $s$ to $milestone_2$ in the neighbors of $w$. The agent, while executing the treasure hunt algorithm, first arrives at $w$, decode the path which is encoded in the neighbors of $w$, returns back to $s$ and then moves according to this learned path.

In order to encode the path from $s$ to $milestone_2$, a similar approach as in Case H can be applied. The sequence of port numbers from $s$ to $milestone_2$  is encoded in the neighbors of $w$. As before, this is done based on the distance from $s$ to $milestone_2$. However, a difficulty arises as $s(0)$ and/or $s(1)$ may also be neighbours of $w$ and there is no way for the agent to learn through which port $w$ is connected to $s(0)$ or $s(1)$. We overcome this difficulty in the following way.

Let $\Gamma$ be the binary string of length $z$ that we want to encode in the neighbors of $w$. Let $N_1(w), N_2(w), \ldots, N_5(w)$ be 5 sets of disjoint neighbors of $w$ and the cardinality of each set is $z$. Since $s(0)$ and/or $s(1)$ may be neighbours of $w$, at least 3 of these 5 neighbour sets of $w$ does not contain either $s(0)$ or $s(1)$. This implies that, if $\Gamma$ is coded in each of these five sets, it is possible to code $\Gamma$ correctly (without interrupting $s(0)$ and $s(1)$ for which pebble placement is already done in Phase 1)
in the nodes of three sets. Hence, while executing the treasure hunt algorithm, if the agent learns all the strings that are coded in these five sets, then $\Gamma$ can be identified as the string that is coded in majority number of sets.
To be specific, the nodes of each of the $N_i(w)$ to encode $\Gamma$, and if $s(0)$ or $s(1)$ appears in this set then skip the corresponding bits of $\Gamma$ while placing pebbles at the time of coding. Another difficulty here is how to make the agent learn where the  set $N_i(w)$ ends and $N_{i+1}(w)$ starts. To overcome this difficulty, instead of simple binary encoding of the sequence of port numbers as explained in Case H, we use a \colb{transformed binary encoding}: replace every `1' by  `11' and every `0' by `10' of $\Gamma$. The advantage of this transformed encoding is that it does not contains the substring 00. We take advantage of this fact to show separation between the sets $N_i(w)$s as follows. First, compute the string $\Gamma$ that has to be encoded in the neighbors of $w$. Compute the transformed encoding $\hat{\Gamma}$. Now, encode $\hat{\Gamma}$ in $N_i(w)$ and next two neighbors after $N_i(w)$ are left blank (no pebbles placed to represent `00'). The agent continues to visit consecutive neighbors of $w$ until it sees two consecutive nodes that does not contains any pebbles. At this point the agent learns that here the set $N_i(w)$ ends. One more difficulty can arise in which the nodes which are left blank may also contain $s(0)$ or $s(1)$. We later show that $\hat{\Gamma}$ still can unambiguously computed by the agent.

The formal description of the placement of pebbles in this case is given below. Let $milestone_2=v_\ell$. Note that as per the definition of milestones, $l \ge 4$.

\begin{enumerate}

\item $\left[\ell=4\right]$: Here the path from $s$ to $milestone_2$ is $v_0,v_1,v_2,v_3,v_4$ and the corresponding sequence of port numbers is $p_0,p_1,p_2,p_3$. Let $\Gamma'= B_0 \cdot B_1 \cdot B_2 \cdot B_3$ and $\Gamma= 11 \cdot \Gamma'$.

\item $\left[\ell=5\right]$: Here the path from $s$ to $milestone_2$ is $v_0,v_1,v_2,v_3,v_4,v_5$ and the corresponding sequence of port numbers is $p_0,p_1,p_2,p_3,p_4$. Let $\Gamma'= B_0 \cdot B_1 \cdot B_2 \cdot B_3\cdot B_4$ and $\Gamma= 10 \cdot \Gamma'$.

\item $\left[\ell=6\right]$: Here the path from $s$ to $milestone_2$ is $v_0,v_1,v_2,v_3,v_4,v_5,v_6$ and the corresponding sequence of port numbers is $p_0,p_1,p_2,p_3,p_4,p_5$. Let $\Gamma'= B_0 \cdot B_1 \cdot B_2 \cdot B_3\cdot B_4 \cdot B_5$ and $\Gamma= 01 \cdot \Gamma'$.

\item $\left[\ell \ge 7\right]$:  Here the path from $s$ to $milestone_2$ is $v_0,v_1, \ldots v_\ell$ and the corresponding sequence of port numbers is $p_0,p_1,\ldots, p_\ell$. In this case, the sequence of port numbers $p_0,p_1,p_2,p_3,p_{\ell-3},p_{\ell-2},p_{\ell-1}$ is coded.
Let $\Gamma'= B_0 \cdot B_1 \cdot B_2 \cdot B_3\cdot B_{\ell-3} \cdot B_{\ell-2} \cdot B_{\ell-1}$ and  $\Gamma= 00 \cdot \Gamma'$.
\end{enumerate}
Let $\hat{\Gamma}$ be the transformed binary encoding of $\Gamma$ and $z$ is the length of $\hat{\Gamma}$.
 Let $N_i(w)$, for $ 1\le i \le 5$, be the set $z$ consecutive neighbors of $w$ starting from the node $w((i-1)(z+2))$. To be specific, $N_i(w) =\{w((i-1)(z+2)), w((i-1)(z+2)+1), \ldots, w(((i-1)(z+2)+z-1)\}$. For each $i$, $1 \le i \le 5$, pebbles are placed at the nodes of $N_i(w)$ as follows. For $1\le a\le z$, if the $a$-th bit of $\hat{\Gamma}$ is 1, then place a pebble at the node $w((i-1)(z+2)+a-1)$ only if $w((i-1)(z+2)+a-1) \not\in \{s(0),s(1)\}$. If $\ell \ge 8$, then place a pebble at each of the nodes $v_{5}, \ldots, v_{\ell-3}$.
 Fig. \ref{fig:case-LH} shows pebble placement in Case L-H.

\begin{figure}[ht!]
    \centering
    \includegraphics[scale=.6]{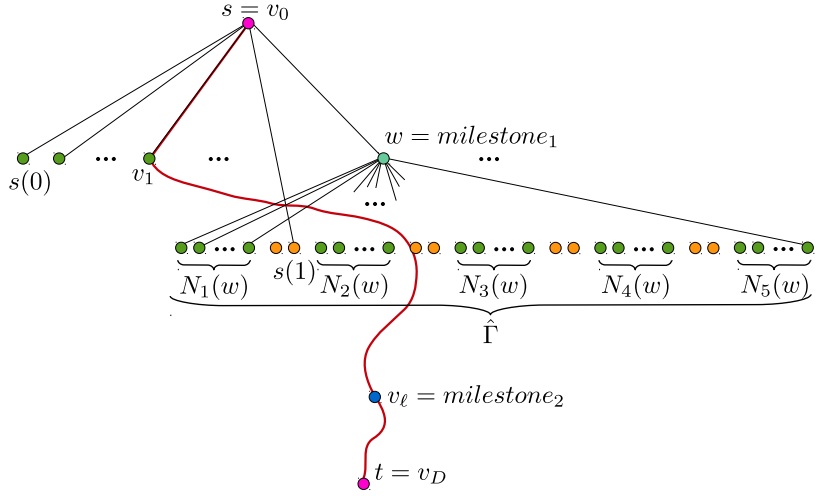}
    \caption{Showing Pebble placement for Case L-H. Here the nodes that are used for pebble placements may be in  $L_0 \cup L_1 \cup L_2$. }
    \label{fig:case-LH}
\end{figure}

\item {\bf Case L-L-H:}
 The placement of pebbles in this case is similar to the placement in Case L-H. The path coded in the neighbors of $w$ is a subpath of the path $P$ starting from $s$.  As before, depending on the position of the  $milestone_2$, different sequences of port numbers are coded in the neighbors of $milestone_1$. Since $milestone_1 \in L_2$ in this case, $s(0)$, $s(1)$ may be connected to $milestone_1$. For this reason, as before, the subpath is coded in five disjoint sets of neighbors of $milestone_1$.

Formal description of placements of pebbles in this case in described below.
 Let $milestone_2=v_\ell$. As per the definition of milestones, $\ell \ge 5$.

\begin{enumerate}

\item $\left[\ell=5\right]$: Here the path from $s$ to $milestone_2$ is $v_0,v_1,v_2,v_3,v_4,v_5$ and the corresponding sequence of port numbers is $p_0,p_1,p_2,p_3,p_4$. Let $\Gamma'= B_0 \cdot B_1 \cdot B_2 \cdot B_3\cdot B_4$ and $\Gamma= 11 \cdot \Gamma'$.

\item $\left[\ell=6\right]$:  Here the path from $s$ to $milestone_2$ is $v_0,v_1,v_2,v_3,v_4,v_5,v_6$ and the corresponding sequence of port numbers is $p_0,p_1,p_2,p_3,p_4,p_5$. Let $\Gamma'= B_0 \cdot B_1 \cdot B_2 \cdot B_3\cdot B_4 \cdot B_5$ and $\Gamma= 10 \cdot \Gamma'$.

\item $\left[\ell=7\right]$: Here the path from $s$ to $milestone_2$ is $v_0,v_1,v_2,v_3,v_4,v_5,v_6,v_7$ and the corresponding sequence of port numbers is $p_0,p_1,\ldots, p_\ell$. In this case, the sequence of port numbers $p_0,p_1,p_2,p_3,p_4,p_5,p_6$ is coded.
Let  $\Gamma'= B_0 \cdot B_1 \cdot B_2 \cdot B_3\cdot B_4 \cdot B_5\cdot B_6$  and  $\Gamma= 01 \cdot \Gamma'$.

\item $\left[\ell \ge 8\right]$:  Here the path from $s$ to $milestone_2$ is $v_0,v_1, \ldots v_\ell$ and the corresponding sequence of port numbers is $p_0,p_1,\ldots, p_\ell$. In this case, the sequence of port numbers $p_0,p_1,p_2,p_3,p_4, p_{\ell-3},p_{\ell-2},p_{\ell-1}$ is coded.
Let $\Gamma'= B_0 \cdot B_1 \cdot B_2 \cdot B_3 \cdot B_4\cdot B_{\ell-3} \cdot B_{\ell-2} \cdot B_{\ell-1}$ and  $\Gamma= 00 \cdot \Gamma'$.
\end{enumerate}

Let $\hat{\Gamma}$ be the transformed binary encoding of $\Gamma$ and $z$ is the length of $\hat{\Gamma}$.
 Let $N_i(w)$, for $ 1\le i \le 5$, be the set $z$ consecutive neighbors of $w$ starting from the node $w((i-1)(z+2))$. To be specific, $N_i(w) =\{w((i-1)(z+2)), w((i-1)(z+2)+1), \ldots, w(((i-1)(z+2)+z-1)\}$. For each $i$, $1 \le i \le 5$, pebbles are placed at the nodes of $N_i(w)$ as follows. For $1\le a\le z$, if the $a$-th bit of $\hat{\Gamma}$ is 1, then place a pebble at the node $w((i-1)(z+2)+a-1)$ only if $w((i-1)(z+2)+a-1) \not\in \{s(0),s(1)\}$. If $\ell \ge 8$, then place a pebble at each of the nodes $v_{5}, \ldots, v_{\ell-3}$.

\begin{remark}\label{rem:1}
 It can be noted that maximum of 8 port numbers must be coded in this case in each of the five disjoint sets of neighbors of $milestone_1$. The transformed binary representation of each port can be of at most $2(1+ \lfloor \log \Delta\rfloor)$ length. Also, two bits are used as marker to represent the distance between first and second milestones. Hence in each set, a string of length $16(1+ \lfloor \log \Delta\rfloor)+2$ is coded. Therefore, over all, among 5 sets, the nodes that are used to code the sequences of port numbers is $80 \lfloor\log \Delta\rfloor+90$.
Also there are two zeros must be coded in between two consecutive sets to show separations between them and at  the end at least 8 nodes kept blank (this is because if $s(1)$ appears in one of the first 4 blank nodes after these 5 sets). Hence the degree of $milestone_1$ must be at least $80 \lfloor\log \Delta\rfloor+106$. Since every milestone is heavy, and $\Delta \ge 2^{10}$ (for $\Delta \ge 2^{10}$, $\Delta \ge 80 \lfloor\log \Delta\rfloor+106$), therefore, such a coding can be done in the neighbors of $milestone_1$.
\end{remark}

\item {\bf Case L-L-L:} This is the case where the two nodes $s(0)$ and $s(1)$ that are used for markers are not connected to $milestone_1$ as $milestone_1 \in L_j$, for some $j \ge 3$. Hence, it is easy to code the paths in the neighbors of $milestone_1$ like Case H. The only difference here is, the subpath coded  is a path starting from $v_{j-3}$. Depending on the position of $milestone_2$, either the sequence of port number corresponding to the entire path from $v_{j-3}$ to $milestone_2$ or the sequence  $p_{j-3},p_{j-2},p_{j-1},p_{j}, p_{j+1},p_{j+2}$, and the last three ports before  $milestone_2$ are coded. For every light nodes  starting from $v_1$ to $v_{j-3}$, a pebble is placed on these nodes that guides the agent towards $milestone_1$. One more difficulty here is the case when $v_1 \in \{s(0),s(1)\}$, as we can not place a pebble in any of the nodes $s(0),s(1)$, otherwise it is not possible to recognize the Case L-L-L.
To overcome this difficulty, we place a pebble at $s$, if $milestone_1 \in L_3$.
      Now, the agent, starting from $s$, first identify the Case L-L-L by not seeing any pebbles in $s(0)$, and $s(1)$. If $s$ contains a pebble, the agent learns that $milestone_1 \in L_3$. It then explores all possible paths of length 3 and finds the node with maximum degree. This node is $milestone_1$. If no pebble is placed at $s$, and  it finds any pebble at a node $s(i)$, $i >1$, then it understands that $s(i)=v_1$. Else, if no pebble is found at any of the nodes $s(2),s(3), \ldots$, and no pebble is present at $s$, then it understands that $v_1 \in \{s(0),s(1)\}$. In this case, the agent first moves to $s(0)$, assuming it is the node $v_1$. If it finds a neighbor of $v_1$, except $s$ that contains a pebble, then it moves to that node (this must be the $v_2$). Else, it moves back to $s$ and then moves to $s(1)$, assuming it as $v_1$. If it finds a neighbor of $v_1$, except $s$ that contains a pebble, then it moves to that node (this must be the node $v_2$). If no pebble is found at any neighbor of $s(1)$ as well, then the agent learns that $milestone_1 \in L_4$. The agent explores all possible paths of length 3 from $s(0)$ and explores all possible paths of length 3 from $s(1)$ and finds the node with maximum degree. The agent moves to this node as this is the node $milestone_1$. The coding in the neighbors of $milestone_1$ is done in the same way as described for Case H.

The formal description of the pebble placement is as follows. Let $j$ be the integer such that $w=milestone_1 \in L_j$ and let $\ell$ be the integer such that $milestone_2=v_\ell$.

\begin{enumerate}

\item $\left[\ell-j=3\right]$: Here the path from $v_{j-3}$ to $milestone_2$ is $v_{j-3},v_{j-2},v_{j-1},v_{j},$ $v_{j+1},v_{j+2},v_{j+3}$ and the corresponding sequence of port numbers is $p_{j-3},$ $p_{j-2},p_{j-1},p_{j},p_{j+1},p_{j+2}$. Let $\Gamma'= B_{j-3} \cdot B_{j-2} \cdot B_{j-1} \cdot B_{j} \cdot B_{j+1} \cdot B_{j+2}$ and $\Gamma= 11 \cdot \Gamma'$.

\item $\left[\ell-j=4\right]$: Here the path from $v_{j-3}$ to $milestone_2$ is $v_{j-3},v_{j-2},v_{j-1},v_{j},$ $v_{j+1},$ $v_{j+2},v_{j+3},v_{j+4}$ and the corresponding sequence of port numbers is $p_{j-3},p_{j-2},p_{j-1},p_{j},p_{j+1},p_{j+2},p_{j+3}$. Let $\Gamma'= B_{j-3} \cdot B_{j-2} \cdot B_{j-1} \cdot B_{j} \cdot B_{j+1} \cdot B_{j+2} \cdot B_{j+3}$ and $\Gamma= 10 \cdot \Gamma'$.
\item $\left[\ell-j=5\right]$: Here the path from $v_{j-3}$ to $milestone_2$ is $v_{j-3},v_{j-2},v_{j-1},v_{j},$ $v_{j+1},v_{j+2},v_{j+3},v_{j+4},v_{j+5}$ and the corresponding sequence of port numbers is $p_{j-3},p_{j-2},p_{j-1},p_{j},p_{j+1},p_{j+2},p_{j+3},p_{j+4}$. Let $\Gamma'= B_{j-3} \cdot B_{j-2} \cdot B_{j-1} \cdot B_{j} \cdot B_{j+1} \cdot B_{j+2} \cdot B_{j+3} \cdot B_{j+4}$ and $\Gamma= 01 \cdot \Gamma'$.

\item$\left[\ell-j\ge 6\right]$: Here the path from $v_{j-3}$ to $milestone_2$ is $v_{j-3},v_{j-2} \ldots v_{\ell}$ and the corresponding sequence of port numbers is $p_{j-3},p_{j-2},\ldots p_{\ell-1}$. In this case, the sequence of port numbers $p_{j-3},p_{j-2},p_{j-1},p_{j},p_{j+1},p_{j+2}, p_{j+3},$ $ p_{\ell-3},p_{\ell-2},p_{\ell-1}$ is coded.
Let $\Gamma'= B_{j-3} \cdot B_{j-2} \cdot B_{j-1} \cdot B_{j} \cdot B_{j+1} \cdot B_{j+2} \cdot B_{j+3} \cdot B_{\ell-3} \cdot B_{\ell-2} \cdot B_{\ell-1}$ and $\Gamma= 00 \cdot \Gamma'$.
\end{enumerate}
Let $\hat{\Gamma}$ be the transformed binary encoding of $\Gamma$ and let $z$ be the length of the string $\hat{\Gamma}$.

For $1 \le a \le z$, place a pebble at $w(a-1)$ if the $a$-th bit of $\hat{\Gamma}$ is 1. If $\ell-j \ge 7$, then place a pebble at each of the nodes $v_{j+4}, \ldots v_{\ell-3}$. Also, for $2 \le i \le j-3$, place a pebble at each of the node $v_i$.  If $j=4$, place a pebble at $s$. If $v_1 \not \in \{s(0),s(1)\}$, then place a pebble at $v_1$.

\end{description}

\noindent {\bf Phase 3: Placing Pebbles to encode paths between other milestones}

The coding of the paths between $milestone_j$ to $milestone_{j+1}$, for $j \ge 2$ are done in the similar fashion as described in Case H. The only difference here is the encoding is done in the neighbors of $milestone_j$ starting from $milestone_j(0)$. For the last milestone, the path from the last milestone to the treasure is coded in the same way as described for the other cases.

Let the total number of milestones be $y$.
For $m =2,3, \ldots, y-1 $, let $v_j$ be the node on $P$ such that $milestone_{m}=v_j$ and $milestone_{m+1}=v_\ell$. For $m=y$, set $\ell=t$.

\begin{enumerate}
\item $\ell-j \le 3$  then $\ell \le j+3$.
The path from $v_j$ to $v_\ell$ is $v_{j},v_{j+1}\ldots v_{\ell}$ and the corresponding sequence of port numbers is $p_{j},p_{j+1}, \ldots,p_{\ell-1}$. Let $\Gamma'= B_j \cdot B_{j+1} \cdots B_{\ell-1}$ and $\Gamma= 11 \cdot \Gamma'$.

 \item $\ell-j=4$  then $\ell=j+4$.
The path from $v_j$ to $v_\ell$ is $v_{j},v_{j+1},v_{j+2},v_{j+3},v_{j+4}$ and the corresponding sequence of port numbers is $p_{j},p_{j+1},p_{j+2},p_{j+2},p_{j+3}$. Let $\Gamma'= B_j \cdot B_{j+1} \cdot B_{j+2}\cdot B_{j+3}$ and $\Gamma= 10 \cdot \Gamma'$.

\item $\ell-j=5$  then $\ell=j+5$.
The path from $v_j$ to $v_\ell$ is $v_{j},v_{j+1},v_{j+2},v_{j+3},v_{j+4},v_{j+5}$ and the corresponding sequence of port numbers is $p_{j},p_{j+1},p_{j+2},p_{j+2},p_{j+3},p_{j+4}$. Let $\Gamma'= B_j \cdot B_{j+1} \cdot B_{j+2} \cdot B_{j+3} \cdot B_{j+4}$ and $\Gamma= 10 \cdot \Gamma'$.

\item If $\ell-j \ge 6 $, and $\ell \ne t$ then let $\Gamma'$ be the binary representation of the sequence of port numbers $p_{j},p_{j+1},p_{j+2},p_{\ell-3},p_{\ell-2},p_{\ell-1}$.  If $\ell=t$ then $\Gamma$ is the binary representation of the sequence of port numbers $p_{j},p_{j+1},p_{j+2}$. $\Gamma=00 \cdot \Gamma'$.
\end{enumerate}

For $1 \le a \le z$, place a pebble at $w(a-1)$ if the $a$-th bit of $\Gamma$ is 1. If $\ell-j \ge 7$,  then place a pebble at each of the nodes $v_{j+4}, \ldots v_{t-3}$. If $\ell=t$, then place a pebble at each of the nodes $v_{t-2}$ and $v_{t-1}$.

\begin{enumerate}
\item $\left[\ell=3\right]$: Here the path from $milestone_1$ to $milestone_2$ is $v_0,v_1,v_2,v_3$ and the corresponding sequence of port numbers is $p_0,p_1,p_2$. Let $\Gamma'= B_0 \cdot B_1 \cdot B_2$ and $\Gamma= 11 \cdot \Gamma'$.

\item $\left[\ell=4\right]$: Here the path from $milestone_1$ to $milestone_2$ is $v_0,v_1,v_2,v_3,v_4$ and the corresponding sequence of port numbers is $p_0,p_1,p_2,p_3$. Let $\Gamma'= B_0 \cdot B_1 \cdot B_2 \cdot B_3$ and $\Gamma= 10 \cdot \Gamma'$.

\item $\left[\ell=5\right]$: Here the path from $milestone_1$ to $milestone_2$ is $v_0,v_1,v_2,v_3,v_4,v_5$ and the corresponding sequence of port numbers is $p_0,p_1,p_2,p_3,p_4$. Let $\Gamma'= B_0 \cdot B_1 \cdot B_2 \cdot B_3 \cdot B_4$ and $\Gamma= 01 \cdot \Gamma'$.

\item $\left[\ell \ge 6\right]$:  Here the path from $milestone_1$ to $milestone_2$ is $v_0,v_1, \ldots v_\ell$ and the corresponding sequence of port numbers is $p_0,p_1,\ldots p_\ell$. As mentioned earlier, the sequence of port numbers $p_0,p_1,p_2,p_{\ell-3},p_{\ell-2},p_{\ell-1}$ is coded in this case.
Let $\Gamma'= B_0 \cdot B_1 \cdot B_2 \cdot B_{\ell-3} \cdot B_{\ell-2} \cdot B_{\ell-1}$ and  $\Gamma= 00 \cdot \Gamma'$.
\end{enumerate}
Let $\hat{\Gamma}$ be the transformed binary encoding of $\Gamma$.
Let $z$ be the length of the string $\hat{\Gamma}$.
For $1\leq i\leq z$, place a pebble on $s(1+i)$, if the $i$-th bit of $\hat{\Gamma}$ is 1.
If $\ell \ge 7$, then place a pebble on each of the nodes $v_{4}, \ldots v_{\ell-3}$.

\subsubsection{Treasure Hunt by The Mobile Agent}

The main idea behind the treasure hunt algorithm executed by the agent is to move from one milestone to the next milestone until the treasure is found. Initially, the agent is at the node $s$. The position of the first milestone is learned by the agent by visiting the two neighbors $s(0)$ and $s(1)$ of $s$. Based on whether a pebble is present in either of these nodes, the agent  determines the location of $milestone_1$ and then moves to the same. According to the pebble placement strategy, the neighbors of every milestone are used to code the path towards the next milestone. From a milestone, the agent visits a set of its neighbors and decode the sequence of port numbers corresponding to the path from the current milestone to the next milestone. Using this information and identifying the pebbles placed at the light nodes along the paths, the agent reaches to the next milestone. This process continues until the treasure is found. The detail description of the treasure hunt algorithm is give below.

The agent follows Algorithm \ref{alg:alg1} to find the treasure. Starting from the node $s$, it first visits the two nodes $s(0)$ and $s(1)$.
If pebbles are found at both the nodes, then the agent follow Algorithm \ref{alg:H}; If a pebble is found at $s(0)$ but no pebble is found at $s(1)$, then the agent follow Algorithm \ref{alg:L-H}; If a pebble is found at $s(1)$ but no pebble is found at $s(0)$, then the agent follows Algorithm \ref{alg:L-L-H}, otherwise the agent follows Algorithm \ref{alg:L-L-L} if no pebbles are found in either of these two nodes.

\begin{algorithm}[ht!]
\SetAlgoLined

{Starting from $s$, the agent visits two nodes $s(0)$, and $s(1)$ one by one and comes back to $s$.}

\If{Both the nodes $s(0)$, and $s(1)$ contains a pebble each}
    {{\textsc{Subroutine\_H}} (Algorithm \ref{alg:H})}
\ElseIf{$s(0)$ contains a pebble and $s(1)$ does not contain any pebble}
    {
    {\textsc{Subroutine\_L-H}} (Algorithm \ref{alg:L-H})
    }
\ElseIf{$s(0)$ does not contain any pebble and $s(1)$ contains a pebble}
    {{\textsc{Subroutine\_L-L-H}} (Algorithm \ref{alg:L-L-H})}

\ElseIf{Neither $s(0)$ nor $s(1)$ contains a pebble }
{{\textsc{Subroutine\_L-L-L} }(Algorithm \ref{alg:L-L-L})}

\caption{\textsc{TreasureHunt}}
\label{alg:alg1}
\end{algorithm}

If two pebbles are found at each of the nodes $s(0)$ and $s(1)$, the agent learns that $s$ is heavy. According to Algorithm \ref{alg:H}, the agent visits the nodes $s(2),s(3),\ldots$ until it found two nodes $s(z+1)$ and $s(z+2)$ such that no pebbles are found in both of these nodes. Let $\hat{\Gamma}=b_2 b_3 \ldots b_z$ be the binary string such that $b_i=1$ if a pebble is found at $s(i)$, else $b_i=0$. Let $\Gamma$ be the string obtained from $\hat{\Gamma}$ by replacing each `11' by '1' and each `10' by 0 of $\hat{\Gamma}$ from left to right, taking two bits a a time. The first two bits of $\Gamma$ represents  the distance of $milestone_2$ from $s$. The agent, knowing the degree of $s$, compute $\alpha_s=1+\lfloor \log deg(s) \rfloor$.  Let  $q_0$ be the integer that is represented by the substring  $\Gamma(3,4+\alpha_s)$. The agent moves along the port $p$ to reach the node $v_1$.  Once it moves to $u_1$, it learns its degree and computes $\alpha_{v_1}$. It then compute the integer $q_1$ that is coded in the substring $\Gamma(5+\alpha_s, 6+\alpha_s+ \alpha_{v_1})$. The agent moves along the port $q_1$ to reach the node $v_2$. The agent continues to move this way distances 3,4,5, if the first two bits represent the markers `11',`10',`01', respectively to reach $milestone_2$. For the marker `00' represented by $s(2)$ and $s(3)$, the agent moves distance 3 as per the above strategy, to reach a node $v_3$. It then visits all the neighbors of $v_3$ and move to the neighbor that contains a pebble. This process continues until for a node, none of its neighbors contains any pebble. In this case, the agent retrieve the next three ports encoded in the rest of the substring of $\Gamma$, one by one and moving to the respective node and move along three edges to reach to $milestone_2$.

\begin{algorithm}[ht!]
\SetAlgoLined

{$CurrentNode=s$.}

{The agent visits the neighbors of $s$ starting from $s(2)$, in the increasing order of the port number through which $s$ is connected to them until it finds two consecutive neighbors where no pebbles are placed.}

{Let $\hat{\Gamma}=b_2 b_3 \ldots b_z$ be the binary string where $b_i$ is 1 if a pebble is found at $s(i)$ and the last two nodes visited in the previous step by the agent are $s(z+1)$ and $s(z+2)$.}

{Let $\Gamma$ be the binary string obtained from $\hat{\Gamma}$ by replacing every `11' by a `1' and every `10' by a `0' from left to right by taking two bits at a time.}

{$CurrentIndex=3$.}

{$MinDistance=3$}

{\textsc{FindNextMilestone($b_2,b_3, \Gamma, MinDistance$)}  (Algorithm \ref{alg:NMS})}

{{\sc Progress}($CurrentNode$) (Algorithm \ref{alg:progress})}
\caption{\textsc{Subroutine\_H}}
\label{alg:H}
\end{algorithm}

If a pebble is found in $s(0)$ but no pebble is found in $s(1)$, the agent learns that $milestone_1 \in L_1$. In this case, it executes \textsc{Subroutine\_L-H}.
The agent visits all the neighbor of $s$ and finds the neighbor $w$ with maximum degree. In case of tie, the agent moves to the node  with maximum degree to which $s$ is connected via smallest port number is . This node $w$ is $milestone_1$. The agent, after moving to $w$ from $s$, starts visiting all the neighbors of $w$ one by one until it finds four consecutive neighbors $w(z),w(z+1)$, $w(z+2)$, and $w(z+3)$ none of which contain any pebble. The agent construct the binary string $\Gamma'=b_0b_1 \ldots b_{z-1}$, where $b_i=1$  if a pebble was found at $w(i)$, else $b_i=0$. This string is split into substrings that are separated by the substring `00' and the agent computes the substring $\hat{\Gamma}$ whose occurrence among these substrings is maximum. Let $\Gamma$ be the string obtained from $\hat{\Gamma}$ by replacing each `11' by `1' and each `10' by `0' of $\hat{\Gamma}$ from left to right, taking two bits a a time.
The agent computes the ports one by one as described for Algorithm \ref{alg:H} and move towards the second milestone. The only difference here is after computing $\Gamma$, the agent comes back to $s$ from $milestone_1$ and the path coded at $\Gamma$ starts from $s$.

\begin{algorithm}[ht!]
\SetAlgoLined
{The agent visits all the neighbors of $s$ and let $w$ be the node that have the maximum degree among all the neighbors of $s$. In case where multiple nodes with maximum degree exists, let $w$ be the node to which $s$ is connected via the smallest port number.}

{The agent moves to $w$. Let $q$ be the incoming port at $w$ of the edge $(s,w)$.}

{The agent visits the neighbors of $w$ in the increasing order of the port number through which $w$ is connected to them until it finds four consecutive neighbors where no pebbles are placed.}

{Go back to $s$ from $w$ using port number $q$. $CurrentNode=s$.}

{Let ${\Gamma'}=b_0 b_1  \ldots b_{z'}$ be the binary string where $b_i$ is 1 if a pebble is found at $w(i)$ and the last four nodes visited in the previous step by the agent are $w(z'), w(z'+1), w(z'+2), w(z'+3)$.}\label{step6}

{Partition $\Gamma'$ into substrings that are separated by two consecutive zeros. Let $\hat{\Gamma}$ be the string that matches with most of these substrings.}

{Let $\Gamma$ be the string obtained from $\hat{\Gamma}$ by replacing every `11' by `1' and every `10' by `0' from left to right by taking two bits at a time. Let $b$ and $b'$ be the first two bits of $\Gamma$. }

{$CurrentIndex=3$. $MinDistance=4$}

{\textsc{FindNextMilestone($b,b', \Gamma, MinDistance$)} (Algorithm \ref{alg:NMS})}

{{\sc Progress}($CurrentNode$) (Algorithm \ref{alg:progress})}
\caption{\textsc{Subroutine\_L-H}}
\label{alg:L-H}
\end{algorithm}

If a pebble is found at $s(0)$ but no pebble at $s(1)$, the agent learns that $milestone_1 \in L_2$. In this case, the agent executes \textsc{Subroutine\_L-L-H} (Algorithm \ref{alg:L-L-H}). From $s$, it explores all possible paths of length 2 from $s$ and moves to the maximum degree node in $L_2$ to which $s$ is connected through the lexicographycally shortest path. After moving to $w$, the agent computes the binary string $\Gamma$ in the same way as described in case of Algorithm \ref{alg:L-H} and proceed towards the second milestone. Here, the string $\Gamma$ that is computed by the agent, codes the path starting from $s$ and the agent, after learning $\Gamma$, moves back to $s$ and moves forward according to subpath coded in $\Gamma$ towards the second milestone.

\begin{algorithm}[ht!]
\SetAlgoLined
{The agent visits all the neighbors of $s$ and let $w$ be the node the maximum degree node in $L_2$ to which $s$ is connected through the lexicographycally shortest path.}

{The agent moves to $w$.}

{ The agent visits the neighbors of $w$ in the increasing order of the port number through which $w$ is connected to them until it finds four consecutive neighbors where no pebbles are placed.}

{Go back to $s$.}

{$CurrentNode=s$.}

{Let $\Gamma'=b_0 b_1 \ldots b_{z'-1}$ be the binary string where $b_i$ is 1 if a pebble is found at $w(i)$ and the last four nodes visited in the previous step by the agent are $w(z'), w(z'+1), w(z'+2), w(z'+3)$.}

{Partition $\Gamma'$ into substrings that are separated by consecutive zeros. Let $\hat{\Gamma}$ be the string that matches with most of these substrings.}

{Let $\Gamma$ be the string obtained from $\hat{\Gamma}$ by replacing every `11' by a `1' and every `10' by a `0' of $\hat{\Gamma}$ by taking two bits at a time from left to right. Let $b,b'$ be the first two bits of $\Gamma$, respectively.}

{$MinDistance=5$. $CurrentIndex=3$}

{\textsc{FindNextMilestone($b,b', \Gamma, MinDistance$)} (Algorithm \ref{alg:NMS})}

{{\sc Progress}($CurrentNode$) (Algorithm \ref{alg:progress})}
\caption{\textsc{Subroutine\_L-L-H}}
\label{alg:L-L-H}
\end{algorithm}

If no pebbles are found at both $s(0)$ and $s(1)$, the agent learns that $milestone_1 \in L_j$ for some $j \ge 3$. It then visits all the neighbors of $s$ and identify the node $v_1$ by finding a neighbor with a pebble. Here, as mentioned in the pebble placement algorithm, the problem occurs when $v_1 \in \{s(0),s(1)\}$, as these two nodes are already used as marker and therefore no pebble can be placed here. In this case, the agent considers both $s(0)$ and $s(1)$ as possible candidates for $v_1$.

Once $milestone_2$ is reached, the agent moves according to Algorithm \ref{alg:progress}. Until the treasure is found, the agent, learn the sequence of port numbers that leads towards the next milestone by visiting the neighbors of the current milestone. Then following this sequence of port numbers and using the pebbles placed on the light nodes, the agent moves to the next milestone. This process continues until the treasure is found.

During the execution of the tresure hunt algorithm, the agent uses a set of global variables, $CurrentNode$, $MinDistance$, and $CurrentIndex$. The variable $CurrentNode$ denotes the node from which the current call of the algorithms are executed. The $MinDistance$ variable stores the integers which is the minimum number of ports that are coded in the neighbor of the current milestone. The $CurrentIndex$ indicates the position of the binary string (represents the sequence of port numbers towards the next milestone) from which the coding of the port number along the shortest path from $CurrentNode$ starts.

The following lemmas ensure the correctness of the proposed algorithm.

\begin{algorithm}
\SetAlgoLined
{$CurrentNode=s$.}

\label{step1}{Visit all the neighbors of $CurrentNode$.}

\If{Treasure is found}{
    Stop and terminate.
}
\If{a pebble found at a node $v$}{
    Move to $v$. Set $CurrentNode=v$. Go to Step \ref{step1}.
}
\Else{
    \If{$CurrentNode \ne s$}{
        \label{step2} Visit all the paths from $CurrentNode$ of length 3. Let $w$ be the node of maximum degree connected to $CurrentNode$ by the lexicographicaly shortest path of length 3.

        Move to $w$. Store the incoming ports of the path from $v$ to $w$ in a stack.
    }
    \Else{
        \If{$s$ contains a pebble}{
            Go to Step \ref{step2}.
        }
        \Else{
            Move to $s(0)$. Let $q$ be the port number of the edge $(s,s(0))$ at $s(0)$. Visit all the neighbors of $s(0)$.

            \If{a pebble is found at a neighbor $v$ of $s(0)$}{
                Move to $v$. Set $CurretNode=v$. Go to Step \ref{step1}.
            }
            \Else{
                Return to $s$ using port $q$ from $s(0)$. Move to $s(1)$. Let $q'$ be the port number of the edge $(s,s(1))$ at $s(1)$. Visit all the neighbors of $s(1)$.

                \If{a pebble is found at a neighbor $v$ of $s(1)$}{
                    Move to $v$. Set $CurretNode=v$. Go to Step \ref{step1}.
                }
                \Else{
                    Move back to $s$. Go to Step \ref{step2}.
                }
            }
        }
    }
}

The agent visits the neighbors of $w$ in the increasing order of the port number through which $w$ is connected to them until it finds two consecutive neighbors where no pebbles are placed.

{Let $\hat{\Gamma}=b_0 b_1 \ldots b_{z-1}$ be the binary string where $b_i$ is 1 if a pebble is found at the node $w(i)$ and the last two nodes visited in the previous step by the agent are $w(z)$ and $w(z+1)$.}

{Let $\Gamma$ be the string obtained from $\hat{\Gamma}$ by replacing every `11' by a `1' and every `10' by a `0' of $\hat{\Gamma}$ by taking two bits at a time from left to right. Let $b,b'$ be the first two bits of $\Gamma$, respectively.}

{Move back to $CurrentNode$ using the path stored in the stack.}

{$CurrentIndex=3$, $MinDiatnace=6$. }

{Algo \textsc{FindNextMilestone($b,b', \Gamma, MinDistance$)}(Algorithm \ref{alg:NMS})}

{{\sc Progress}($CurrentNode$) (Algorithm \ref{alg:progress})}
\caption{\textsc{Subroutine\_L-L-L}}
\label{alg:L-L-L}
\end{algorithm}

\begin{algorithm}[ht!]
\SetAlgoLined
{Let  $p$ be the integer that is represented by the substring constructed from $\Gamma(i,x)$.}

{The agent move from the current node to the node $u$ to which the current node is connected via port $p$.}

\If{Treasure found}{
    Stop and terminate.
}
\Else{
    $CurrentNode=u$. $CurrentIndex=x+1$
}

\caption{\textsc{Movement($i$,$x$, $\Gamma$)}}
\label{alg:movement}
\end{algorithm}

\begin{algorithm}[ht!]
\SetAlgoLined

\If{$a=1$ and $b=1$}{
    \For{$j\gets1$ \KwTo $i$}{
        \textsc{Movement($CurrentIndex$,$\alpha_{CurrentNode}$, $\Gamma$)} (Algorithm \ref{alg:movement})
    }
}
\Else{
    \If{$a=1$ and $b=0$}{
        \For{$j=1$ to $i+1$}{
            \textsc{Movement($CurrentIndex$,$\alpha_{CurrentNode}$, $\Gamma$)}
        }
    }
    \Else{
        \If{$a=0$ and $b=1$}{
            \For{$j=1$ to $i+2$}{
                \textsc{Movement($CurrentIndex$,$\alpha_{CurrentNode}$, $\Gamma$)}
            }
        }
        \Else { \label{step3}
            \For{$j=1$ to $i$}{
                \textsc{Movement($CurrentIndex$,$\alpha_{CurrentNode}$, $\Gamma$)}
            }
            \While{\label{step5} a pebble is found in some neighbor of $CurrentNode$}{
                Move to the neighbor $u$ of $CurrentNode$ that contains a pebble. $CurrentNode=u$
            }
            \For{$i=1$ to 3}{
                \label{step4} \textsc{Movement($CurrentIndex$,$\alpha_{CurrentNode}$, $\Gamma$)}
            }
        }
    }
}

{\textsc{Progress($CurrentNode$)} (Algorithm \ref{alg:progress})}

\caption{\textsc{FindNextMilestone($a,b, \Gamma, i$)}}
\label{alg:NMS}
\end{algorithm}

\begin{algorithm}[ht!]
\SetAlgoLined

{The agent visits the neighbors of $CurrentNode$ starting from $CurrentNode(0)$, in the increasing order of the port number through which $CurrentNode$ is connected to them until it finds two consecutive neighbors where no pebbles are placed.}

{Let $\hat{\Gamma}=b_0b_1 \ldots b_z$  be the binary string where $b_i$ is 1 if a pebble is found at the node connected to $s$ through port $i$ and the last two nodes visited in the previous step by the agent are $s(z+1)$ and $s(z+2)$.}

{Let $\Gamma$ be the binary string obtained from $\hat{\Gamma}$ by replacing every 11 by a 1 and every 10 by a 0 from left to right by taking two bits at a time.}

{Let $b,b'$ be the first two bits of $\Gamma$.}

{$CurrentIndex=3$. $MinDistance=3$}

{\textsc{FindNextMilestone($b,b', \Gamma, MinDistance$)} (Algorithm \ref{alg:NMS})}

{{\sc Progress}($CurrentNode$) (Algorithm \ref{alg:progress})}
\caption{\textsc{Progress($CurrentNode$)}}
\label{alg:progress}
\end{algorithm}

\begin{lemma}
In time $O((dist(s,milestone_1)\log \Delta+\log ^3\Delta))$,  the agent successfully reaches  to $milestone_1$ starting from the node $s$.
\label{lem:milestone}
\end{lemma}
\begin{proof}
We prove the lemma for each of the cases, Case H, Case L-H, Case L-L-H, and Case L-L-L, one by one.

\begin{itemize}
\item {\bf Case H:} $s$ is itself $milestone_1$, and therefore the lemma is trivially true.
\item {\bf Case L-H:} $milestone_1 \in L_1$.  As per the definition of $milestone_1$, here the maximum degree node in $L_1$ (in the case of tie, the node to which $s$ is connected by the edge with minimum port number) is $milestone_1$. According to the placement of pebbles, a pebble is placed at $s(0)$ and no pebble is at $s(1)$. According to Algorithm \ref{alg:alg1}, the agent, finding a pebble at $s(0)$ and no pebble at $s(1)$  executes Algorithm \ref{alg:L-H}. It moves to the node whose degree is maximum among all neighbors of $s$ and to which $s$ is connected via the smallest port number. Hence, the agent, following Algorithm \ref{alg:L-H}, reaches $milestone_1$. Note that the time to reach $milestone_1$ is $O(\log \Delta)$, as $s$ is a light node, and therefore, $deg(s)$ in $O(\log \Delta)$.

\item {\bf Case L-L-H:} $milestone_1 \in L_2$.  In this case, as per the definition of $milestone_1$, the maximum degree node in $L_2$  (in the case of tie, the node to which $s$ is connected by a path of length 2 which is lexicographically shortest) is $milestone_1$.
According to the placement of pebbles, a pebble is placed at $s(1)$, and no pebble is at $s(0)$. According to Algorithm \ref{alg:alg1}, the agent, finding a pebble at $s(1)$ and no pebble at $s(0)$, executes Algorithm \ref{alg:L-H}. It moves to the node whose degree is maximum among all nodes which are connected to $s$ by a path of length 2 (in the case of tie, the node to which $s$ is connected by a path of length 2 which is lexicographically shortest). Hence, the agent, by executing Algorithm \ref{alg:L-L-H}, reaches $milestone_1$. Note that the time to reach $milestone_1$ is $O(\log^2 \Delta)$, as $s$ and all the nodes in $L_1$ are light nodes, and therefore the total number of paths from $s$ of length 2 is $O(\log ^2\Delta)$.

\item {\bf Case L-L-L:} $milestone_1 \in L_j$, for some $j \ge 3$. As per the pebble placement strategy, no pebbles are placed on either of $s(0)$ and $s(1)$. Further, if $milestone_1 \in L_3$, then the node with maximum degree in $L_3$ is chosen as $milestone_1$. No pebble is placed in any of the nodes in $L_1$ and $L_2$. Now, after not founding any pebble at $s(0)$ and $s(1)$, if there is a pebble in $s$, then the agent learns that $milestone_2 \in L_3$. According to Algorithm \ref{alg:L-L-L}, the agent explores all possible paths of length 3 and moves to the node with maximum degree (tie is broken as earlier). This ensures that it reaches to $milestone_1$.

If $milestone_1 \in L_j$, for $j \ge 3$, then a pebble is placed at each of the nodes $v_2, \ldots, v_{j-3}$ and a pebble is placed at $v_1$ if $v_1 \not \in \{s(0),s(1)\}$. As per the definition of $milestone_1$ in this case, the node in $L_j$ with maximum degree and at a distance 3 from $v_{j-3}$ is defined as $milestone_1$.
According to the pebble placement algorithm, no pebbles are placed at $s(0)$ and $s(1)$. If $v_1 \not \in \{s(0),s(1)\}$, then a pebble is placed at $v_1$. The agent, according to Algorithm \ref{alg:L-L-L}, starting from $s$ visits all the neighbor of $s$, and moves to $v_1$, which is the node that contains a pebble.
Then from $v_1$, it moves to $v_2$, as this is the only neighbor of $v_1$ that contains a pebble, and continue this process until it reaches to $v_{j-3}$, none of whose neighbors contains any pebbles (except $v_{j-4}$ but this the agent can detect by storing the incoming ports every time).
Therefore, by Step \ref{step2} of Algorithm \ref{alg:L-L-L}, the agent explores all possible paths of length 3 from $v_{j-3}$ and moves to the node that have maximum degree at distance 3 from $v_{j-3}$.
 Hence, the agent successfully reaches to $milestone_1$ in this case as well. The only case remains to show is that when $milestone_1 \in L_j$, for $j \ge 3$ and $v_1 \not \in \{s(0),s(1)\}$. In this case, both of the nodes $s(0),s(1)$ are possible position for the node $v_1$ and the agent explores as per the other case assuming $s(0)$ as $v_1$ once and then $s(1)$ as $v_1$ next. Hence, by a similar argument as mentioned in the previous cases (Cases H, L-H, L-L-H), the agent reaches to $milestone_1$ successfully. Next, we compute the time taken by the agent from $s$ to $milestone_1$. Note that each of the nodes $s, v_1, v_2, \ldots, v_{j-3}$ are light nodes in this case. Since, a pebble is placed at each of these nodes and the agent moves to these nodes one by one by exploring all the neighbors and finding the pebble, the total time taken by the agent is $O((j-3)\log \Delta)$, i.e, $O(dist(s,milestone_1)\log \Delta)$. Once the agent reaches to $v_{j-3}$, it finds the $milestone_1$ by exploring all possible paths of length 3, that will take $O(\log^3) \Delta$-time. Therefore, the total time taken to reach $milestone_1$ is $O((dist(s,milestone_1)\log \Delta+\log ^3\Delta))$.
\end{itemize}
\end{proof}

In the next two lemmas, we prove that from $milestone_1$ the agent reaches to $milestone_2$ for each of the different cases. The first lemma proves this for the Case H and Case L-L-L and the second lemma proves this for the Case L-H and Case L-L-H.

\begin{lemma}
The agent successfully reaches to $milestone_2$ form $milestone_1$ for Case H and Case L-L-L.
\end{lemma}
\begin{proof}
The proofs for Case H and Case L-L-L are similar. The only difference is that the path coded in the neighbors of $milestone_1$ starts from $milestone_1$ itself in Case H, whereas, in Case L-L-L, the path coded in the neighbors of $v_j=milestone_1$ starts from $v_{j-3}$. We first describe the proof for Case H.

In Case H, recall that $s$ is heavy and therefore $milestone_1=s$. Suppose that $milestone_2 \in L_j$. According to Algorithm \ref{alg:H}, the agent, visits all the neighbors of $s$ one by one, starting from $s(2)$, in the increasing order of the port number until it finds two consecutive neighbors without pebbles. It then computes the binary string  $\hat{\Gamma}=b_2 \cdots b_z$, where $b_i=1$ if a pebble is found in $s(i)$, else $b_i=0$. Then the string $\Gamma$ is constructed from $\hat{\Gamma}$ by replacing every `11' by a `1' and every `10' by `0' from left to right of $\hat{\Gamma}$ by taking two bits at a time.
As per the pebble placement algorithm, $\Gamma$ represents a sequence of port numbers that leads towards $milestone_2$. To be specific, depending on the position of $milestone_2$, the following two cases may happen.

\begin{itemize}
    \item {\boldmath{\bf $\left[milestone_2 \in L_j,~ j\le 6 \right]$:}} In this case, the sequence $p_0,p_1, \ldots, p_j$ is encoded using the pebbles, i.e., $\Gamma'$ is the binary representation of the sequence $p_0,p_1, \ldots, p_j$, where $\Gamma=b_2b_3 \cdot \Gamma'$ and one of $b_2$ and $b_3$ is 1. The agent identify this by finding at least one pebble at  $s(2)$ or $s(3)$. After learning $\Gamma'$, and knowing the degree of $s$, the agent compute the substring containing first $1+ \lfloor \log degree(s)\rfloor$ bits and decode the integer $p_0$. After decoding $p_0$, it moves to the node $v_1$ from $s$ by taking the port $p_0$. As soon as the agent reaches $v_1$, it learns the degree of $v_1$ and then decode the integer $p_1$ which is decoded in the next $1+\lfloor \log deg(v_1) \rfloor$ bits of $\Gamma'$. Following this port from $v_1$, the agent reaches $v_2$ and learns the degree of $v_2$. Continuing this way, the agent decodes the integers $p_2, \ldots,p_{j-1}$ one by one and finally reaches to node $v_j$ which is $milestone_2$.
    \item {\boldmath{\bf $\left[milestone_2 \in L_j, \text{for~} j > 6 \right]$:}} Here the sequence $p_0,p_1,p_2,p_{j-3},p_{j-2},p_{j-1}$ is coded in $\Gamma'$. According to the pebble placement algorithm, there are no pebble at $s(2)$ and $s(3)$. The agent identify this case by not finding any pebbles both at $s(2)$ and $s(3)$. According to Algorithm \ref{alg:H}, the agent set $CurrentIndex=4$ and $Mindistance=3$, and call the subroutine {\textsc{FindNextMilestone$(0,0,\Gamma,MinDistance)$}}. Since $b_2=0$ and $b_3=0$, according to Steps \ref{step3} - \ref{step4}  of  \textsc{FindNextMilestone}, the agent decodes the port numbers $p_0,p_1,p_2$ one by one and then reaches $v_1$, $v_2$, and $v_3$ respectively. Once it reaches $v_3$, it starts exploring all the neighbors of $CurrentNode$ and moves to the node that contains a pebble and update $CurrentNode$ (Ref. Step \ref{step5}). It continues to move in this way until it finds no pebble in the neighborhood of $CurrentNode$ (except the node from where it reaches $CurrentNode$). According to the pebble placement algorithm, this situation arises when $v_{j-3}$ is the $CurrentNode$. From $v_{j-3}$, the agent decodes $p_{j-3}$, from $\Gamma$, and moves to $v_{j-2}$, and decodes $p_{j-2}$ from $\Gamma$ and moves to $v_{j-1}$, and then finally, decodes $p_{j-1}$ and reach to $v_j$ which is $milestone_2$.
\end{itemize}
In Case L-L-L, the proof is similar as in Case H. After learning the sequence $\Gamma$ which is encoded in the neighbors of $milestone_1$, the agent returns back to $v_{j-3}$ and move one by one along $P$ by decoding the port numbers $p_{j-3},p_{j-2}, \ldots$, to reach $milestone_2$. In case $milestone_2$ more than 6 distance apart from $v_j$, the agent first moves to $v_{j+3}$ from $v_{j-3}$ by decoding first 6 ports one by one. After that it's movement is guided by the light nodes from $v_{j+4}$ to $v_{\ell-3}$, where $milestone_2=v_\ell$. Once it reaches to $v_{\ell-3}$, the agent finds no pebble in the neighborhood of $CurrentNode$ which is $v_{\ell-3}$ (except the node from where it reaches $CurrentNode$, i.e., $v_{\ell-4}$).
From $v_{\ell-3}$, the agent decodes $p_{\ell-3}$ from $\Gamma$, and moves to $v_{\ell-2}$, and decodes $p_{\ell-2}$ from $\Gamma$ and moves to $v_{\ell-1}$, and then finally, decodes $p_{\ell-1}$ from $\Gamma$ and reaches to $v_\ell$ which is $milestone_2$.
\label{lem:h}\end{proof}

\begin{lemma}
The agent successfully reaches to $milestone_2$ from $milestone_1$ for Case L-H and Case L-L-H.
\end{lemma}

\begin{proof}
We prove the lemma for Case L-H. The proof for Case L-L-H is similar.

According to Lemma \ref{lem:milestone}, the agent successfully reaches to the node $w =milestone_1$. After reaching $milestone_1$, the agent computes the binary string $\hat{\Gamma}$ (Ref. Step \ref{step6} of Algorithm \ref{alg:L-H}). According to the pebble placement algorithm, the sequence of port numbers corresponding to the path from $s$ to $milestone_2$ is coded in 5 disjoint sets of neighbors of $milestone_1$. We show that the agent correctly computes this encoding of the path. In this case, a pebble is placed in $s(0)$ and no pebble is at $s(1)$. Consider the following cases.

\begin{itemize}
  \item $s(0), s(1) \not \in \{w(0), w(1), \ldots, w(5z+11)\}$. In this case, $\Gamma$ is the string that is equals to $\hat{\Gamma} 00 \hat{\Gamma} 00 \hat{\Gamma} 00 \hat{\Gamma} 00 \hat{\Gamma}$. Hence, the string separated by `00' which occurs most of the time is $\hat{\Gamma}$. Hence, $\hat{\Gamma}$ in learned correctly by the agent in this case.
  \item $s(0) \in N_i(w)$ and $s(1) \in N_j(w)$ for some $i,j \le 5$. In this case, the string $\hat{\Gamma}$ can not be correctly coded in the sets $N_i(w)$ and $N_j(w)$. To be more specific,  one bit of the string encoded in $N_i(w)$ and  one bit of the string encoded in $N_j(w)$ may differ from $\hat{\Gamma}$. (in case $i=j$, at most two bits can differ from $\Gamma$). Therefore, In this case, the string $\Gamma'$ which is computed by the agent  is the string that is equals to $\Gamma_1 00 \Gamma_2 00 \Gamma_3 00 \Gamma_4 00 \Gamma_5$, where $\Gamma_i, \Gamma_j \ne \hat{\Gamma}$ and $\Gamma_k= \hat{\Gamma}$, for $k \in \{1,2,3,4,5\}\setminus \{i,j\}$. It can be seen that $\hat{\Gamma}$ is the most occurring substring separated by 00, and hence $\hat{\Gamma}$ is correctly computed in this case.
  \item $s(0) \in\{ w(i(z+2)-2),w(i(z+2)-1)$, for some $1 \le i \le 4$, i.e, $s(0)$ (it contains a pebble in this case) is a node that are left blank to represent two consecutive zeros. Without loss of generality, suppose that $s(0) =w(z)$. In this case $\Gamma'$ is the string that looks like $\hat{\Gamma} 10 \hat{\Gamma} 00 \hat{\Gamma} 00 \hat{\Gamma} 00 \hat{\Gamma}$. The substrings that are separated by `00' are $\Gamma_1$, and three substrings equals to $\hat{\Gamma}$, where $\Gamma_1=\hat{\Gamma}10\hat{\Gamma}$. Hence, $\hat{\Gamma}$ is the substring that is the most occurred substring and therefore, $\hat{\Gamma}$ is correctly computed in this case as well.
  \item $s(0) \in \{w(5z+8), w(5z+9), w(5z+10), w(5z+11)\}$. This is the case where the node $s(0)$ appears as one of the four nodes after all the nodes of the five sets $N_i(w)$. Here, the string $\Gamma'$ looks like $\hat{\Gamma} 00 \hat{\Gamma} 00 \hat{\Gamma} 00 \hat{\Gamma} 00 {\Gamma_1}$ where $\Gamma_1$ is a substring whose one of the four end bits corresponds to $s(0)$. In this case as well, $\hat{\Gamma}$ is the most occurred string and therefore the agent correctly learn $\hat{\Gamma}$.
  \end{itemize}

  Once the agent correctly learn $\hat{\Gamma}$, it computes $\Gamma$ by replacing every `11' by `1' and every `10' by `0'. If the first two bits of $\Gamma$ is the substring `11', or `10', or `01' then the agent learns that $milestone_2$ is at  4 or 5, or 6 distance apart, respectively from $s$. The agent returns back to $s$ from $w$ and move one edge at a time by decoding the integers $p_0$, $p_1, p_2, \ldots, p_\ell$, for $\ell=4,5,6$ and then successfully reaches to $milestone_2$. If the first two bits of $\Gamma$ is the substring `00', then agent learns that $milestone_2$ is at least 7 distance apart. The agent returns back to $s$ from $w$, and move one edge at a time by decoding the integers $p_0$, $p_1$, $p_2$, and $p_3$, to reach the node $v_4$. According to the pebble placement algorithm, if $milestone_2 \in L_\ell$, then no pebble is placed on any node in $L_{\ell-2}$, and a pebble is placed at each of the nodes $v_5, \ldots, v_{\ell-3}$. The agent, from $v_4$ moves to $v_5$, then to $v_6$, and so on until the node $v_{\ell-3}$ by moving to the neighbor of the current node that contains a pebble. Once it reaches to $v_{\ell-3}$, it can not find any other neighbor of $v_{\ell-3}$, other than $v_{\ell-4}$ that contains a pebble. The agent learns in this point that it reaches to $v_{\ell-3}$ and then it starts decoding the rest three port numbers one by one and following them, to finally reach to $milestone_2$.

\end{proof}

\begin{lemma}

After reaching $milestone_j$, for some $j \ge 2$, the agent successfully either reaches to  $milestone_{j+1}$, if exists, or finds the treasure.
\end{lemma}
\begin{proof}
We prove this lemma using induction. By Lemma \ref{lem:milestone}, the agent successfully reaches $milestone_1$.  To prove the base case, we prove that from $milestone_1$, it reaches to $milestone_2$ successfully.

Let $y$ be the total number of milestones. Suppose that the agent successfully reached to $milestone_j$ for $j<y$. It then executes Algorithm \ref{alg:progress} (subroutine \textsc{Progress}) according to which it visits all the neighbors of $CurrentNode$ which is $milestone_j$ and compute the string $\Gamma'$.
As per the pebble placement algorithm, $\Gamma'$ represents a sequence of port numbers that leads towards $milestone_{j+1}$. To be specific, depending on the position of $milestone_{j+1}$, the following cases may happen.

\begin{itemize}
    \item If $milestone_j \in L_\ell$ and $milestone_{j+1} \in L_{\ell'}$, such that $\ell'- \ell \le 6$. In this case, the sequence $p_\ell,p_{\ell+1}, \ldots, p_{\ell'-1}$ is encoded in the neighbors of, i.e., $\Gamma'$ is the binary representation of the sequence $p_0,p_1, \ldots, p_j$, where $\Gamma=b_1b_1 \cdot \Gamma'$ and one of $b_1$ and $b_2$ is 1. The agent identify this by finding at least one pebbles at  $milestone_j(0)$ and $milestone_j(1)$. After learning $\Gamma'$, and knowing the degree of $milestone_j$, the agent compute the substring containing first $1+ \lfloor \log degree(s)\rfloor$ bits and decode the integer $p_\ell$. After decoding $p_\ell$, it moves to the node $v_{\ell+1}$ from $v_\ell$ by taking the port $p_\ell$. As soon as the agent reaches $v_{\ell+1}$, it learns the degree of $v_{\ell+1}$ and then decode the integer $p_{\ell+1}$ which is decoded in the next $1+\lfloor \log deg(v_{\ell+1}) \rfloor$ bits of $\Gamma'$. Following this port from $v_{\ell+1}$, the agent reaches $v_{\ell+2}$ and learns the degree of $v_{\ell+2}$. Continuing in this way the agent finally reaches to node $v_{\ell'}$ which is $milestone_2$.

    \item If $milestone_{j+1} \in L_{\ell'}$, for $ \ell'-\ell> 6$.  In this case, the sequence $p_{\ell}$, $p_{\ell+1}$, $p_{\ell+2}$, $p_{\ell'-3}$, $p_{\ell'-2}$, $p_{\ell'-1}$ is coded in $\Gamma'$. According to the pebble placement algorithm, there are no pebble at $milestone_j(0)$ and $milestone_j(1)$. The agent identify this case by not finding any pebbles at $milestone_j(0)$, and $milestone_j(1)$. According to Algorithm \ref{alg:progress}, the agent set $CurrentIndex=3$ and $Mindistance=3$, and call  {\textsc{FindNextMilestone$(0,0,\Gamma,MinDistance)$}} (Algorithm \ref{alg:NMS}). According to Step \ref{step3}- \ref{step4}  of  Algorithm \ref{alg:NMS}, the agent decodes the port numbers $p_\ell,p_{\ell+1},p_{\ell+2}$ one by one and then reaches $v_{\ell+1}$, $v_{\ell+2}$, and $v_{\ell+3}$, respectively. Once it reaches $v_{\ell+3}$, it starts exploring all the neighbors of $CurrentNode$ and moves to the node that contains a pebble and update $CurrentNode$ (Ref. Step \ref{step5} of Algorithm \ref{alg:NMS}). It continues to move in this way until it finds no pebble in the neighborhood of $CurrentNode$ (except the node from where it reaches $CurrentNode$). According to the pebble placement algorithm, this situation arises when $v_{\ell'-3}$ is the $CurrentNode$. From $v_{\ell'-3}$, the agent decodes $p_{\ell'-3}$ from $\Gamma$, and moves to $v_{\ell'-2}$, and decodes $p_{\ell'-2}$ from $\Gamma$ and moves to $v_{\ell'-1}$, and then finally, decodes $p_{\ell'-1}$ and reach to $v_{\ell'}$ which is $milestone_{j+1}$.
\end{itemize}

Hence using the induction hypothesis, the agent reaches to $milestone_y$. In the neighbors of $milestone_y$, the path from $milestone_y$ to the treasure is coded. If the distance of $t$ from $milestone_y$ is $\le 6$, then the entire path to the treasure is coded in the neighbors of $milestone_y$. The agent, in a similar fashion as described earlier, learn this path by decoding the port numbers one by and one and finally finds the treasure. In case the path to the treasure is of length at least 7, then only first three ports from $milestone_1$ are coded in the neighbors of $milestone_y$. The agent from $milestone_y=v_\ell$ reached to the node $v_{\ell+3}$,  and then it starts exploring all the neighbors of $CurrentNode=v_{\ell+3}$ and moves to the node that contains a pebble and update $CurrentNode$ of Algorithm \ref{alg:NMS}). It continues to move this way and finally finds the treasure and stop.
\end{proof}

We now present our final result in the following theorem.
\begin{theorem}
The agent finds the treasure in $O(D \log \Delta+\log^3 \Delta)$-time.
\end{theorem}
\begin{proof}
There are $O(\log \Delta)$ pebbles that are placed in the neighbors of each milestone and at most one pebble is placed on each of the light nodes on $P$. Since there can be at most $\frac{D }{3}$ milestones and $O(D)$ light nodes, hence the total number of pebbles used is $O(D \log \Delta)$.

The agent visits $O(\log \Delta)$ neighbors of each milestone and all the neighbors of each light node on $P$.  The total time taken to find the treasure from $milestone_1$ is $D(\log \Delta)$. Also, by Lemma \ref{lem:milestone},  the agent reaches to $milestone_1$ in $O(dist(s,milestone_1) \log \Delta+\log^3 \Delta)$-time. Since $dist(s,milestone_1) \le D$, the time for treasure hunt is $O(D \log \Delta + \log^3\Delta)$.
\end{proof}

\section{Lower bound}\label{sec:lb1}

In this section, we show a a lower bound $\Omega(D \log \Delta)$ for time of treasure hunt. To be specific, we construct a class of instances of treasure hunt such that if the time for treasure hunt is `short', then any algorithm using any number of pebbles can not find the treasure within this short time for some instances.

Let $T$ be a complete tree of height $D$ where the degree of the root $r$ and each internal node is $\Delta$. There are $\Delta \cdot (\Delta-1)^{D-1}$ leaves in $T$. Let $p= \Delta \cdot (\Delta-1)^{D-1}$ and  $u_1, \ldots, u_p$ be the leaves of $T$ in lexicographical ordering of the shortest path from the root $r$. For $1 \le i\le p $, we construct an input $B_i$ as follows. The tree $T$ is taken as the input graph, $r$ as the starting point of the agent, and $u_i$ as the position of the treasure.  Let $\cB$ be the set of all inputs $B_i$, $1 \le i\le p $.

Suppose that there exists an algorithm $\cA$ that can solve the treasure hunt in time $t$  for all graphs of maximum degree $\Delta$ and $D$ being the distance between the initial position of the agent and the treasure. The movement of the agent according to  $\cA$ can be viewed as a sequence of port numbers $p_1,p_2, \ldots,p_{t}$, where the values of the port numbers depend only on the placement of pebbles. To be specific, at any time $t'\le t$, the agent takes a port $p$ if a pebble is placed in the current vertex, else it takes a port $p'$. Note that $p$ and $p'$ may be the same port. In other words, depending on the placement of the pebbles, at any step, the agent can take one of the two possible ports, and therefore total at most $2^{{t}}$ possible sequences of port numbers the agent may follow from its initial position to find the treasure. This observation is formally proved in the following lemma.

\begin{lemma}\label{lem:lem1}
For any treasure hunt algorithm $\cA$ taking $t$-time, there are at most $2^{t}$ possible sequences of port numbers the agent may follow for the treasure hunt.
\end{lemma}

\begin{proof}
We use induction to prove the above statement. Consider the execution of the algorithm for $t=1$, i.e., the agent traverse exactly one edge. There are two possible cases. If no pebble is placed at the starting position of the agent, the agent must take some specified port  $p$, following $\cA$. If a pebble is placed at the starting node, then the agent must take some specified port $q$ (may be $p=q$). Hence, there are at most one of the two different sequences  of port numbers can be followed by the agent, according to $\cA$. Suppose that, the statement is true for any $t'<t$, i.e., there are at most $2^{t'}$ possible different sequences of port numbers, one which the agent might follow depending on the position of the pebbles. Let $Q$ be a sequence of port numbers among these $2^{t'}$ different sequences. Let $v$ be the position of the agent after following $Q$ starting from the initial position. As argued, in the case for $t=1$, there are two possible port numbers one of which the agent can take in time $t'+1$ depending on whether a pebble is placed or not at the current node $v$. Since this statement is true for every possible sequences of port numbers of length $t'$, this implies that there are at most $2 \times 2^{t'}= 2^{t'+1}$ many possible sequences of port numbers of length $t'+1$. Therefore, the statement is followed by induction.
\end{proof}

The following theorem proves the lower bound result.

\begin{theorem}

There exists a graph $G$ of maximum degree $\Delta~(\ge 2)$ and diameter $D~(\ge 3)$ such that any deterministic algorithm must require $\Omega(D \log \Delta)$-time for the treasure hunt irrespective of the number of pebbles placed on the nodes of $G$.
\end{theorem}
\begin{proof}
Suppose that there exists an algorithm that solves the treasure hunt in time at most $\frac{D \log (\Delta-1)}{2}$ for the class of inputs $\cB$. By Lemma \ref{lem:lem1}, there are at most $2^{ \frac{D\log (\Delta-1)}{2}}=(\Delta-1)^{\frac{D}{2}}$ different possible sequences of port numbers the agent can follows. Since $|\cB| =\Delta (\Delta-1)^{D-1}$, therefore, by the Pigeon hole principle, there exist at least $(\Delta-1)^{\frac{D}{2}}$ inputs in $\cB$ for which the agent follows the same sequence of port numbers. Since $\frac{D \log (\Delta-1)}{2} < (\Delta-1)^{\frac{D}{2}}$, for $\Delta \ge 3$, it is not possible to reach the treasure for $(\Delta-1)^{\frac{D}{2}}$ many inputs using the same sequence of port numbers of length $\frac{D \log (\Delta-1)}{2}$. This contradicts the fact that $\cA$ solves the treasure hunt in time at most $\frac{D \log (\Delta-1)}{2}$ for the class of inputs $\cB$. Hence the theorem follows.
\end{proof}

\section{Conclusion}
We propose an algorithm for the treasure hunt problem that finds the treasure in an anonymous graph in $O(D \log \Delta+\log^3 \Delta)$-time. We also prove a lower bound of $\Omega(D \log \Delta)$. Clearly, there is a small gap between the upper and lower bounds, however, the gap is smaller than any polynomial of $\Delta$. A natural open question is to find tight upper and lower bounds for the problem. Another interesting problem is to study trade-off between number of pebble and time for treasure hunt in anonymous networks.

%
%
%
 \bibliographystyle{plain}
\bibliography{main}
\end{document}